\documentclass[pra,twocolumn,floatfix,superscriptaddress,notitlepage,longbibliography]{revtex4-1}

\usepackage[normalem]{ulem}
\usepackage{graphicx, color, graphpap}
\usepackage{enumitem}
\usepackage{amssymb}
\usepackage{amsthm}
\usepackage{multirow}
\usepackage[dvipsnames]{xcolor}
\definecolor{DarkPurple}{RGB}{119,30,125}
\usepackage[colorlinks=true,urlcolor=DarkPurple,citecolor=DarkPurple,linkcolor=DarkPurple]{hyperref}
\usepackage[T1]{fontenc}
\usepackage{algorithm}
\usepackage{algpseudocode}
\usepackage{verbatim}
\usepackage{mathtools}
\usepackage{titlesec}
\usepackage{float}
\usepackage{booktabs}
\usepackage{cellspace}
\usepackage{soul}
\usepackage{multirow}

\long\def\ca#1\cb{} 

\newcommand{\braket}[2]{\langle #1 \hspace{1pt} | \hspace{1pt} #2 \rangle}

\newcommand{\ketbra}[2]{| \hspace{1pt} #1 \rangle \langle #2 \hspace{1pt} |}

\newcommand{\ket}[1]{|#1\rangle}               
\newcommand{\bra}[1]{\langle #1|}              

\newcommand{\ip}[2]{\langle #1|#2\rangle}      

\usepackage{circuitikz}
\usepackage{tikz}
\usetikzlibrary{arrows, shapes.gates.logic.US, calc}
\usetikzlibrary{arrows,decorations.pathreplacing}
\usetikzlibrary{backgrounds,fit,decorations.pathreplacing,calc}

\newcommand{\CC}{\mathcal{C}}
\newcommand{\DC}{\mathcal{D}}

\newcommand{\HC}{\mathcal{H}}

\newcommand{\PC}{\mathcal{P}}

\renewcommand{\geq}{\geqslant}
\renewcommand{\leq}{\leqslant}

\newcommand*{\id}{\openone}

{}
{}
\newtheorem{corollary}{Corollary}
\newtheorem{theorem}{Theorem}
\newtheorem{lemma}{Lemma}

\newtheorem{proposition}{Proposition}

\newtheorem{definition}{Definition}

\newcommand{\tr}[1]{\textnormal{tr}\left[#1\right]}
\DeclareMathOperator{\E}{\mathbb{E}}

\renewcommand{\bm}[1]{\boldsymbol{#1}}            
\newcommand{\ind}[1]{\textnormal{ind}[#1]}        
\newcommand{\cov}{\textnormal{Cov}}
\newcommand{\var}{\textnormal{Var}}
\newcommand{\med}{\textnormal{median}}

\begin{document}
\title{Local measurement strategies for multipartite entanglement quantification}

\author{Luke Coffman}
\affiliation{Department of Physics, University of Colorado, Boulder, Colorado 80309, USA}
\affiliation{JILA, NIST and University of Colorado, Boulder, Colorado 80309, USA}

\author{Akshay Seshadri}
\affiliation{Department of Physics, University of Colorado, Boulder, Colorado 80309, USA}
\affiliation{JILA, NIST and University of Colorado, Boulder, Colorado 80309, USA}

\author{Graeme Smith}
\affiliation{Department of Physics, University of Colorado, Boulder, Colorado 80309, USA}
\affiliation{JILA, NIST and University of Colorado, Boulder, Colorado 80309, USA}
\affiliation{Institute for Quantum Computing and Department of Applied Math, University of Waterloo}

\author{Jacob L. Beckey}
\affiliation{Department of Physics, University of Colorado, Boulder, Colorado 80309, USA}
\affiliation{JILA, NIST and University of Colorado, Boulder, Colorado 80309, USA}

\begin{abstract} 
Despite multipartite entanglement being a global property of a quantum state, a number of recent works have made it clear that it can be quantified using only local measurements. This is appealing because local measurements are the easiest to implement on current quantum hardware. However, it remains an open question what protocol one should use in order to minimize the resources required to estimate multipartite entanglement from local measurements alone. In this work, we construct and compare several estimators of multipartite entanglement based solely on the data from local measurements. We first construct statistical estimators for a broad family of entanglement measures using local randomized measurement (LRM) data before providing a general criterion for the construction of such estimators in terms of projective 2-designs. Importantly, this allows us to de-randomize the multipartite estimation protocol based on LRMs. In particular, we show how local symmetric informationally complete POVMs enable multipartite entanglement quantification with only a single measurement setting. For all estimators, we provide both the classical post-processing cost and rigorous performance guarantees in the form of analytical upper bounds on the number of measurements needed to estimate the measures to any desired precision.
\end{abstract}
\maketitle

\section{Introduction}
In the nearly 90 years since the famous Einstein-Podolsky-Rosen (EPR) paper~\cite{einstein1935Can} initiated the study of entanglement, there has been an enormous effort to understand this phenomenon theoretically \cite{horodecki2009quantum} and, as quantum technologies have matured, to probe it experimentally \cite{guhne2009Entanglement,islam2015measuring, kaufman2016quantum, bluvstein2022quantum}. While many questions in entanglement theory remain only partially resolved, the experimental violations of Bell's inequalities \cite{bell1964Einstein,clauser1969Proposed} over the past several decades \cite{freedman1972Experimental,hensen2015Loopholefree,giustina2015SignificantLoopholeFree,shalm2015Strong} imply, conclusively, that entanglement can lead to fundamentally \textit{non-local} correlations. Nonetheless, the entanglement content of an unknown quantum state can be quantified and characterized using \textit{local} measurements alone \cite{elben2018renyi,brydges2019probing,elben2019statistical,notarnicola2021randomized,elben2020MixedState,vermersch2023Manybody}.

The most informative measurements allowable in quantum mechanics correspond to positive operator-valued measures (POVMs) on many identical copies of ones quantum state (i.e.~on $\rho^{\otimes n}$), so-called many-copy measurements. Because all single-copy POVMs are a subset of all POVMs, many-copy measurements must be at least powerful as single-copy measurements. The same argument holds when one restricts further to \textit{local} measurements on the individual subsystems. There is a natural trade-off between ease of experimental implementation and information gained. This intuition can be made rigorous in a number of ways, including comparing lower bounds on the sample complexity of various learning tasks under different measurement restrictions. 

For example, consider the canonical learning task of quantum state tomography (QST), where one's goal is to produce an estimate $\hat{\rho}$ of an unknown $d$-dimensional state $\rho$, such that $\|\hat{\rho} - \rho \|_1 < \epsilon$ with high probability. If restricting to local measurements, one requires at least $\Omega(\frac{d^3}{\epsilon^2})$ copies of the state \cite{haah2017SampleOptimal}, while with many-copy measurements, $\Theta (\frac{d^2}{\epsilon^2})$ copies are necessary and sufficient \cite{odonnell2016Efficient,haah2017SampleOptimal}. Because $d$ scales exponentially with the number of quantum systems, this seemingly minor improvement would have significant impact in practice. While many-copy measurements are necessary for optimal sample complexity \cite{bubeck2020Entanglement}, they are nowhere near being implementable on today's quantum systems. Even global single-copy measurements are infeasible for moderate system sizes. Thus, in this work, we focus on \textit{local} measurements on each individual qubit of an $n$-qubit state $\rho$. 

Requiring full knowledge of state, as in QST, is almost always superfluous. Typically, we are interested in functions of quantum states that can be estimated directly without ever estimating the quantum state itself. This realization has led to a burgeoning field of study regarding the optimal estimation of properties of quantum states, in which the estimation procedures have lower sample complexity and classical post-processing requirements than full state tomography \cite{vanenk2012measuring,odonnell2015Quantum,elben2018renyi,elben2019statistical,rath2021Quantum,huang2020predicting,elben2020MixedState,huang2022quantum,elben2023randomized,vermersch2023Manybody}. Our work contributes to the ongoing effort to devise practical methods for the experimental estimation of entanglement using easily implementable measurements.

In this work, we construct statistical estimators of a broad family of entanglement measures based only on local POVMs. While it is known that entangled measurements on many identical copies of a state are necessary for the optimal property testing~\cite{bubeck2020Entanglement} and that there are fundamental trade-offs between entanglement characterization and detection with and without multi-copy measurements \cite{lu2016Tomography,liu2022characterizing}, separations between resource requirements for the estimation of multipartite entanglement measures remain an interesting area of study with many open problems. Moreover, as these are the simplest measurements to implement in practice, our results are relevant to experiments, both current and near-term, that wish to probe entanglement in quantum systems without the need to prepare identical copies simultaneously or utilize multi-copy measurements.

Specifically, we generalize Ref.~\cite{ohnemus2023Quantifying}, which devised estimators for the generalized concurrence \cite{carvalho2004decoherence} based on local randomized measurements (LRMs), to a more general family of multipartite entanglement measures called the concentratable entanglements (CEs) \cite{beckey2021computable}. We provide analytical performance guarantees on the estimation of this family of entanglement measures in the form of upper bounds on the number of measurements needed to estimate the CEs to $\epsilon$ precision with high probability (i.e. upper bounds on the sample complexity). Moreover, we provide estimates of the post-processing cost of computing our estimators, a crucial consideration in practice. We then provide a distinct estimation procedure, based on LRM data and using median-of-means (MoM) estimation, that provides a square-root enhancement in the sample complexity of this task -- making multipartite entanglement quantification in systems of several tens of qubits a feasible prospect.

A remaining limitation of LRMs is that they require an exponential number of measurement \textit{settings} to be implemented, which can be experimentally challenging. To address this, we provide a general theorem that shows that any projective $2$-design can be used to construct an estimator for the CEs. In particular, this theorem implies that local symmetric informationally complete (SIC) POVMs can be used to estimate all of the CEs using a \textit{single measurement setting}, generalizing the work in Ref.~\cite{stricker2022Experimental} to the study of multipartite entanglement. While the classical post-processing of this method is more costly than the other methods, it is likely preferable to having to change the experimental measurement setting an exponential number of times. Thus, we expect these methods to be of interest to the experimental community.

\section{Preliminaries}\label{sec:preliminaries}
In this section, we establish our notation and provide some essential facts from quantum information theory and classical statistics that are needed to prove our main results. To begin, let $\mathcal{B} = \{\ket{j}\}$ denote a basis of a finite-dimensional Hilbert space $\mathcal{H}$. Then, $\{\ket{j} \ket{j'} \mid \ket{j}, \ket{j'} \in \mathcal{B}\}$ forms a basis for the composite space $\mathcal{H} \otimes \mathcal{H}$. The SWAP operator $\mathbb{F}: \mathcal{H}\otimes \mathcal{H} \rightarrow \mathcal{H}\otimes \mathcal{H}$, as the name suggests, is the operator that swaps two tensor components
\begin{align}
     \mathbb{F}\ket{j}\ket{j'} = \ket{j'}\ket{j}, \quad \forall~ \ket{j},\ket{j'} \in \mathcal{B}. 
\end{align}
This operator is diagonal in the Bell basis and has eigenvalues $\pm 1$. The $+1$ and $-1$ eigenspaces are called the symmetric and anti-symmetric subspaces, respectively. Denoting the symmetric (anti-symmetric) subspaces as $\Pi_+$ ($\Pi_-$), we can express the spectral decomposition of the SWAP operator as $\mathbb{F} = \Pi_+ - \Pi_-$. Moreover, these projectors resolve the identity $\mathbb{I} \otimes \mathbb{I} = \Pi_+ + \Pi_-$, allowing us to express the subspace projectors as 
\begin{align}
    \Pi_{\pm} = \frac{\mathbb{I} \otimes \mathbb{I} \pm \mathbb{F}}{2}.
\end{align}
The multipartite entanglement measures we consider in this paper are functions of the subsystem purities, so in the proof of our main results, we will utilize the following well-known relationship between the SWAP operator and a state's purity, called the SWAP ``trick.''

\begin{lemma}[The SWAP ``trick''] \label{lemma:swap-trick}
For an $n$-qubit state $\rho$, the following equality holds
\begin{align}
    \tr{\mathbb{F}\rho^{\otimes 2}} &= \tr{\rho^2}.
\end{align}
\end{lemma}
This lemma can be proven from the definitions of the SWAP operator and the trace (see, for example, Appendix A of Ref.~\cite{beckey2023Multipartite}). Around the turn of this century, it was realized that Lemma~\ref{lemma:swap-trick} suggests a method of purity estimation if one could prepare two identical copies of the quantum state and coherently manipulate the systems with high fidelity \cite{ekert2002Direct}. The necessary level of coherent control has only become possible in the past decade on ion trap and neutral atom systems \cite{islam2015measuring,kaufman2016quantum,bluvstein2022quantum}, but remains at the cutting-edge. In an effort to avoid the use of two-copy measurements, Ref.~\cite{vanenk2012measuring} initiated the study of purity estimation using single-copy measurements and many great theoretical and experimental works to this end have been produced \cite{elben2018renyi,elben2019statistical,stricker2022Experimental}. While the current work focuses on the task of estimating multipartite entanglement measures using only local measurements, we suspect our techniques could be adapted to the direct estimation of all non-linear functionals of the form $\tr{\rho^k}$ for $k \in \mathbb{Z}^+$.

\subsection{Concentratable Entanglement}
In this work, we consider the estimation of the family of multipartite entanglement measures defined in Ref.~\cite{beckey2021computable}. This family includes the pure state entanglement measures from Refs.~\cite{meyer2002global,brennen2003observable,carvalho2004decoherence} as special cases and has several properties that make them interesting (see Ref.~\cite{beckey2021computable} for details). That work shows how to compute the measures using a parallelized controlled-SWAP circuit, which requires two identical copies of an $n$-qubit state and $n$ ancillary qubits. In a follow-up work \cite{beckey2023Multipartite}, the requirements were decreased to two identical copies of the state of interest using Bell basis measurements. Still, creating two identical copies of a moderately-sized quantum state and acting coherently on it is at the cutting edge of quantum information processing \cite{bluvstein2022quantum,bluvstein2023Logical}. 

Before constructing estimators of the CEs, we must introduce the quantities themselves. To that end, we now define the CEs and mention some useful facts about them. 
\begin{definition}[Concentratable Entanglement~\cite{beckey2021computable}] \label{def:CEs}
Let $\ket{\psi}$ be a pure quantum state of $n$ qubits and let $[n]=\{1,2,\ldots,n\}$ denote the full set of qubit labels. For any non-empty set of qubit labels $S \in \PC([n]) \setminus \{\emptyset\}$, the Concentratable Entanglement (CE) is defined as 
\begin{align}\label{eq:CE}
    \CC_{\ket{\psi}}(S) &= 1 -\frac{1}{2^{s}} \sum_{\alpha \in \PC(S)} \tr{\rho_{\alpha}^2},
\end{align}
where $s:=|S|$ denotes the cardinality of $S$, and the $\rho_{\alpha}$'s are reduced states of $\ket{\psi}\bra{\psi}$ obtained by tracing out subsystems with labels not in $\alpha$. For the trivial subset, we define $\tr{\rho_{\emptyset}^2}:=1$.
\end{definition}

Crucial to the present work is the following expression for the CEs \cite{beckey2023Multipartite}
\begin{align}\label{eq:CE-local-symmetric-subspaces}
    \CC_{\ket{\psi}}(S) = 1 - \tr{\rho^{\otimes 2} \prod_{i \in S} \Pi_+^i},
\end{align}
where $\Pi_+^i := \mathbb{I}_1 \otimes \dotsm  \otimes  \mathbb{I}_{i-1}  \otimes \Pi_+^i \otimes \mathbb{I}_{i+1} \otimes \dotsm \otimes \mathbb{I}_n $ denotes the symmetric subspace projector on the $i$-th qubit of the first and second copy of $\rho$ and identities for all $j \in S$ such that $j \neq i$. Note this convention for the representation of operators is used throughout.

Importantly, when $S=[n]$, the CE is related to the multipartite concurrence defined in Ref.~\cite{carvalho2004decoherence}, denoted $c_n (\ket{\psi})$, via the simple expression 
\begin{align}\label{eq:CE-concurrence}
    \mathcal{C}_{\ket{\psi}}([n]) = \frac{c_n (\ket{\psi})^2}{4},
\end{align}
which is the measure Ref.~\cite{ohnemus2023Quantifying} consider estimating using local random measurements. Our results will hold for that measure as well as all others obtained by choice of any $S \subseteq [n]$. While there are many interesting formulas for the CEs given in Refs.~\cite{beckey2021computable,beckey2023Multipartite}, these will suffice to appreciate our main results. We now turn to some facts from classical statistics that will be needed to prove our sample complexity upper bounds.

\subsection{Results from Classical Statistics}\label{sec:classical-stats}
In this work, we will work exclusively with bounded random variables $X_1, \dots, X_M$, such that each $X_i \in [a,b]$ and $a, b \in \mathbb{R}$.
We will denote the expectation value of a random variable $X$ by $\mathbb{E}[X]$ and its variance by $\var(X)$.
We denote the covariance of two (possibly dependent) random variables $X$ and $Y$ by $\cov[X, Y] = (X - \E[X]) (Y - \E[Y])$.
We denote estimators of statistical quantities with hats. For example, if the actual parameter of interest is denoted $\theta$, we will denote an estimator of the parameter as $\hat{\theta}$. We say an estimator is unbiased if $\E[\hat{\theta}] =\theta$. Many of our estimators will be expressed using indicator functions, which are defined as
\begin{align}
\ind{A}=
    \begin{cases}
        1,& \mathrm{ condition } ~A~ \mathrm{ is}~\mathrm{ true }\\
        0,&  \mathrm{otherwise}.
    \end{cases}
\end{align}
For example, given a fixed, length $s$ bitstring $\bm{z}_0 \in \{0, 1\}^{s}$, $\ind{\bm{Z} = \bm{z}_0}$ takes the value $1$ when the random variable $\bm{Z}$ is equal to $\bm{z}_0$, and takes the value $0$ otherwise. Crucially, the probability of obtaining a particular bitstring, $\bm{z}_0$ can be expressed as $\E[\ind{\bm{Z}=\bm{z}_0}] = P(\bm{z}_0).$ Note that we will use uppercase letters to denote random variables, and the corresponding lowercase letter to denote a specific instance of the random variable. 

In practice, finite sample statistics prohibit us from exactly determining quantities of interest. A natural approach is then to ask how many measurements are sufficient to guarantee that our estimate is close to the true value with high probability. This is a well-studied problem in classical statistics and results in so-called \textit{sample complexity upper bounds}. We now present the two methods that will allow us to provide analytical performance guarantees on our estimators of multipartite entanglement.

We start with a well-known result called Hoeffding's inequality, which provides exponential concentration for sums of independent bounded random variables. We will state it without proof as it is a standard result proven in most mathematical statistics textbooks (e.g. Sec. 2.2 of Ref.~\cite{vershynin2018HighDimensional}).

\begin{proposition}[Hoeffding’s inequality] \label{fact:hoeffding}
Let $X_1, \dots, X_M$ be independent random variables such that $a \leq X_i \leq b$ and $\E[X_i]=\mu$ for all $i \in [M]$.
Then, given a precision $\epsilon > 0$ and a confidence level $1 - \delta \in (0, 1)$, choosing $M = \lceil (b - a)^2 \log(2/\delta) / (2 \epsilon^2) \rceil$ suffices to guarantee that
\begin{align}
    \label{eq:Hoeffding-mean}
    \mathbb{P}\left( \left|\frac{1}{M}\sum_{i=1}^M X_i - \mu \right| \geq \epsilon \right) \leq \delta.
\end{align}
\end{proposition}
If, however, the random variables in question have a large allowed range (but small variance), Hoeffding will often be looser than bounds that utilize the information about this second moment. In other words, the number of measurements required to achieve $\epsilon$-precision will be overestimated. To provide stronger concentration, we will make use of MoM estimation (e.g. Sec. 2.3 of Ref.~\cite{lerasle2019Lecture}). 

\begin{proposition}[Median-of-means (MoM) estimation]\label{prop:med-of-means}
Suppose that $X_1, \dotsc, X_M$ are i.i.d. random variables, with variance bounded above by $\sigma^2 > 0$.
Let $\epsilon > 0$ be the precision and $1 - \delta$ be the confidence level, and $M = N_B B$ be the total number of samples,
where the number of samples per batch is $N_B = \lceil 8 \log(1/\delta) \rceil$ and the number of batches is $B = \lceil 4 \sigma^2 / \epsilon^2 \rceil$.
Let $\widehat{\mu}_b$ denote the empirical mean of $X_{(b - 1) B + 1}, \dotsc, X_{b B}$, for $b \in \{1, \dotsc, N_B\}$.
Then, we have
\begin{align}
\mathbb{P}\left(\left|\med(\widehat{\mu}_1, \dotsc, \widehat{\mu}_B) - \mathbb{E}[X]\right| \geq \epsilon\right) \leq \delta.
\end{align}
\end{proposition}
Both of these results allow us to derive analytical performance guarantees in the form of upper bounds on the sample complexity of an estimator. With these preliminaries in mind, we turn now to our main results.
\begingroup
\begin{table*}[ht!]
    \centering
    \setlength{\tabcolsep}{6pt}
    \renewcommand{\arraystretch}{1.5}
    \begin{tabular}{c|c|c|c}
        \multirow{2}{2cm}{\centering Estimator Name} & \multirow{2}{3cm}{\centering Sample Complexity Upper Bound}  & \multirow{2}{3cm}{\centering Classical Post-processing} & \multirow{2}{3cm}{\centering \# of Measurement Settings} \\
        & & & \\
        \hline
        LRM-Mean, $K=2$ & $M_1 = O\left( \left(\frac{9}{4}\right)^{s} \log\left(\frac{1}{\delta}\right) \epsilon^{-2} \right)$ & $O( K^2 s M_1)$ & $O\left( \left(\frac{9}{4}\right)^{s} \log\left(\frac{1}{\delta}\right) \epsilon^{-2} \right)$  \\
        LRM-MoM, $K=2$ & $M_2 = O\left( \left(\frac{3}{2}\right)^{s} \log\left(\frac{1}{\delta}\right) \epsilon^{-2}\right)$
        & $O( K^2 s M_2)$ & $O\left( \left(\frac{3}{2}\right)^{s} \log\left(\frac{1}{\delta}\right) \epsilon^{-2}\right)$ \\
        SIC-MoM, $K=2$ & $M_3 = O\left(3^{s} \log\left(\frac{1}{\delta}\right) \epsilon^{-2}\right)$ & $O(K^2 s M_3)$ & 1  \\
        SIC-MoM, $K=K_{\mathrm{opt}}$ &  $M_4 = O\left( \left[\left(\frac{3}{2}\right)^{s} \epsilon^{-2} + \sqrt{3^{s}} \epsilon^{-1}\right] \log\left(\frac{1}{\delta}\right) \right)$ &  $\min \{ O(K_{\mathrm{opt}} 2^{s} M_4),  O(K_{\mathrm{opt}}^2 s M_4) \}$ & 1 \\
    \end{tabular}
    \caption{\textbf{Summary of Local Strategies for Multipartite Entanglement Quantification.} The second column of this table represents asymptotic upper bounds on the number of measurements needed to estimate the CEs to precision $\epsilon$, with probability at least $1-\delta$. The third gives asymptotic estimates of the classical time complexity, where we assume that all elementary operations are $O(1)$. Finally, the fourth column summarizes the number of measurement settings needed for each estimator.}
    \label{table:complexities}
\end{table*}
\endgroup

\section{Main Results}

We first generalize and provide analytical performance guarantees for the estimators of multipartite entanglement via LRMs constructed in Ref.~\cite{ohnemus2023Quantifying}. We then utilize median-of-means estimation to obtain a (quadratically) better upper bound on the sample complexity CE estimation using LRM data. Then, we show that any local POVM that forms a projective $2$-design can be used to construct an estimator for the CEs. A corollary of this result is the de-randomization of the LRM protocol which enables multipartite entanglement quantification using a \textit{single measurement setting}. All of our sample complexity upper bounds, as well as estimates of the worst case classical post-processing cost, are summarized in Table~\ref{table:complexities}.

\subsection{CEs via LRMs}

As mentioned above, randomized measurements have been studied in many estimation tasks, but most relevant to our work is the recent work of Ref.~\cite{ohnemus2023Quantifying} in which the authors present a method of estimating the multipartite concurrence \cite{carvalho2004decoherence} using local randomized measurement data. Because one recovers the generalized concurrence easily from the CEs (see Eq.~\eqref{eq:CE-concurrence}), Ref.~\cite{ohnemus2023Quantifying} is easily generalizable to the entire family of CEs. 

As depicted in Fig.~\ref{fig:LRM}, the LRM protocol entails preparing the state of interest, applying  $U = \prod_{i = 1}^n U_i$, where each $U_i$ is taken from some single qubit matrix distribution (e.g. Haar distribution or single-qubit Cliffords \cite{mele2023Introduction}), to our state and then measuring in the computational basis $\{\ket{\mathbf{z}}\}$. Note that  This yields a bitstring $\mathbf{z}$, where $\boldsymbol{z}:=z_1 \cdot z_2 \cdot \dotsm z_n$, with $z_i \in \{0,1\}$, and where the probability of obtaining $\mathbf{z}$ is given as $P(\mathbf{z}) = \tr{U\rho U^{\dagger} \ketbra{\mathbf{z}}{\mathbf{z}}}$. Note that if one only cares about a subset $S \subseteq [n]$ of qubits, then one can restrict the unitary rotation and subsequent computational basis measurement to the subsystems defined by $S$. Specifically, this corresponds to rotation by the unitary $U = \prod_{i \in S} U_i$, by which mean $U_i \in U(2)$ is a unitary acting on the $i$th qubit for $i \in S$, and we implicitly assume that the action of the unitary is trivial (i.e., $\id$) on qubits not in the set $S$. We then have the following proposition.

\begin{proposition}[CEs via LRMs]\label{prop:CE-via-LRM-exact}
    Let $\rho = \ketbra{\psi}{\psi}$ be an $n$-qubit pure state, $U=\prod_{i \in S}U_i$ for $U_{i}\in U(2)$ be the tensor product of single-qubit Haar random unitaries in $S\subseteq [n]$, and $P_{U}(\mathbf{z})$ be the probability of measuring bit-string $\mathbf{z} \in \{0, 1\}^{s}$.
    The CEs can be obtained exactly via 
    \begin{align}\label{eq:CE-LRM-exact-single-string}
        \mathcal{C}_{\ket{\psi}}(S)=1-3^s \E_U[P_{U}(\mathbf{z})^2]
    \end{align}
    for any bitstring $\mathbf{z}$, where $\E_U[\cdot]$ represents an average over the Haar measure.
\end{proposition}
The proof of this theorem relies on the fact that
\begin{align}
    \mathbb{E}_{U}[U^{\otimes 2}\ketbra{\mathbf{z}}{\mathbf{z}}U^{\dagger^{\otimes 2}}] &=\frac{1}{3}\Pi_{+},
\end{align}
which follows from Schur's lemma and Eq.~\eqref{eq:CE-local-symmetric-subspaces}. We include a detailed proof in Appendix~\ref{app:proof-Ohnemus-LRM-Hoeffding}. Crucially, we emphasize that this holds for \textit{any} fixed bitstring $\mathbf{z} \in \{0,1\}^{s}$. Because this expression is equally valid for all bitstrings, it follows that we can simply take a uniform average over all bitstrings
\begin{align}\label{eq:CE-LRM-exact-all-strings}
        \mathcal{C}_{\ket{\psi}}(S)=1-\left(\frac{3}{2}\right)^{s} \sum_{\mathbf{z} \in \{0,1\}^{s}}\E_U[P_{U}(\mathbf{z})^2],
\end{align}
which is a crucial experimental consideration because the probability of obtaining any particular bitstring decays exponentially with the number of qubits being probed. Moreover, taking such an average enables improved statistical estimators for two key reasons. First, as we will see below, the uniform average allows us to reduce the constant factor that appears in the sample complexity from $3^s$ to $\left(\frac{3}{2}\right)^s$. And second, the uniform average allows us to focus on estimating the sum of squared probabilities, rather than the squared probabilities themselves. As we will see, this is exponentially more efficient in terms of classical post-processing for many cases of interest.

\begin{figure}[ht!]
    \centering
    \includegraphics[width=0.45\textwidth]{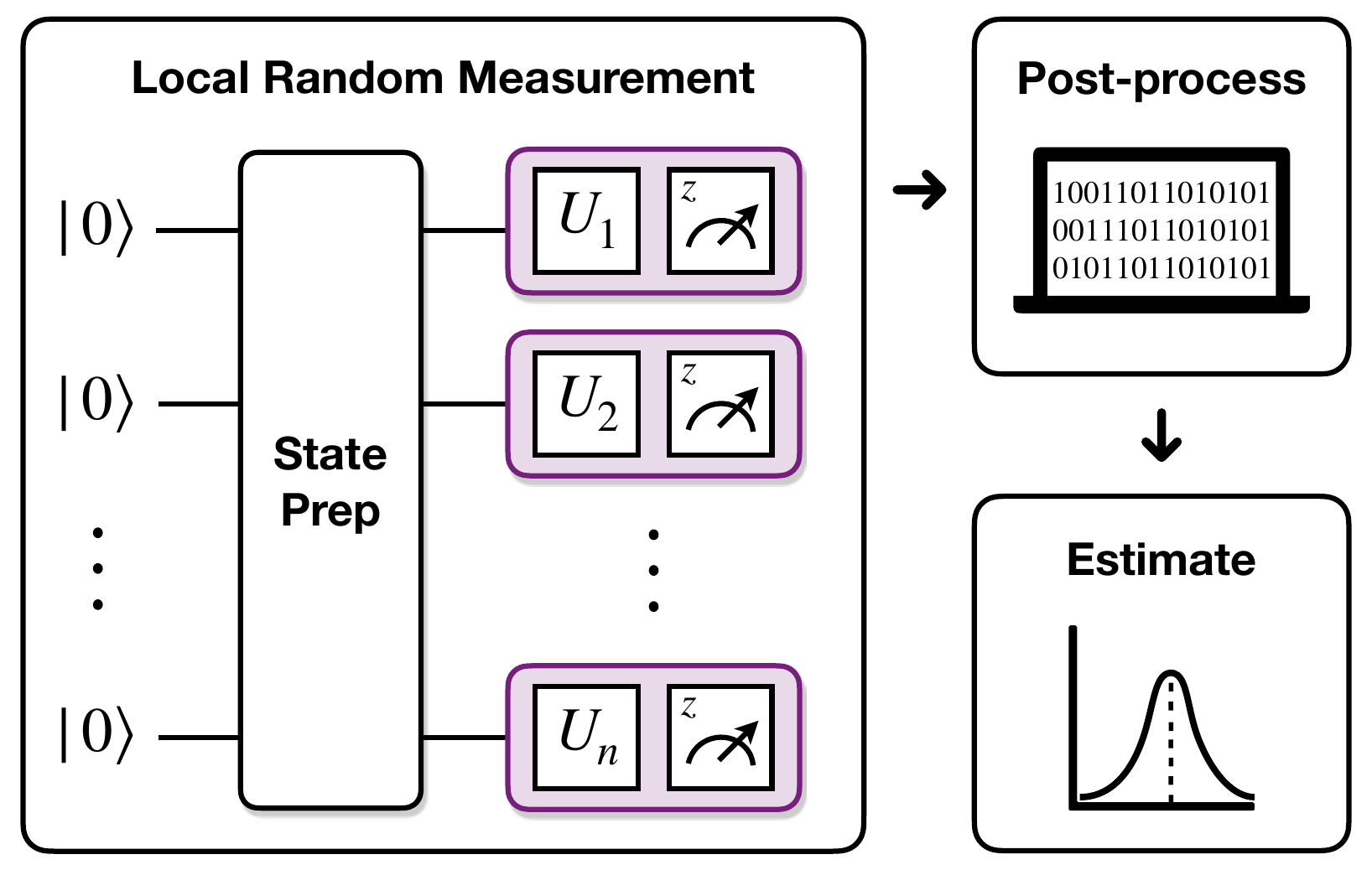}
    \caption{\textbf{Local randomized measurement.} A local randomized measurement (LRM) is simply a random local unitary on each subsystem, followed by a computational basis measurement.}
    \label{fig:LRM}
\end{figure}

In practice, one cannot compute this average exactly over the Haar measure, but instead samples $L$ unitaries from the Haar measure (or a unitary design of the appropriate degree \cite{mele2023Introduction}) and approximates the exact average. Moreover, Eq.~\eqref{eq:CE-LRM-exact-single-string} depends on the probability, at a fixed unitary, of obtaining the bitstring $\mathbf{z}$. This too, cannot be computed exactly, but is estimated by repeating, for each random unitary, the computational basis measurement $K$ times and taking the sample average. This yields a total measurement budget of $M=LK$. Because of these two sources of finite-sampling error, providing convergence guarantees for estimators based on LRM data can be challenging, with many groups resorting to various approximations, special cases, or numerical simulations \cite{brydges2019probing,elben2018renyi,elben2019statistical,vermersch2018unitary,ohnemus2023Quantifying}. While these methods are still informative, it would be preferable to obtain analytical performance guarantees. Before showing how to obtain such guarantees, let us understand the methods used in Ref.~\cite{ohnemus2023Quantifying}, where they construct unbiased estimators of $P^2_U(\mathbf{z})$ to then approximate the actual quantity of interest $\E_U[P_{U}(\mathbf{z})^2]$.

To understand their result, let $\bm{z} \in \{0,1\}^{s}$ denote a bitstring resulting from an LRM as depicted in Fig.~\ref{fig:LRM}.
Ohnemus \textit{et al.} \cite{ohnemus2023Quantifying} show that
\begin{align} \label{eq:p-squared-unbiased}
     \widehat{P}^{(2)}_U({\bm{z}}) &= \widehat{P}_U({\bm{z}})\frac{(K \widehat{P}_U({\bm{z}}) - 1)}{K - 1},
\end{align}
is an unbiased estimator of the squared probability $P_U(\bm{z})^2$, for a fixed unitary $U$.
Here, $\widehat{P}_U({\bm{z}})=\frac{1}{K}\sum_{k=1}^K \ind{\bm{Z}_k=\bm{z}}$ denotes the fraction of the $K$ outcomes equal to $\bm{z}$.
However, the classical run-time of estimating each probability is $O(K)$, and there are $2^{s}$ terms in the sum that appears in Eq.~\eqref{eq:CE-LRM-exact-all-strings}. Thus, the classical post-processing scales exponentially in the number of subsystems being probed, i.e. $O(K2^{s})$. Moreover, proving upper bounds on the sample complexity of the CEs is difficult if one estimates the squared probabilities. Hoeffding's inequality and the union bound give an analytical sample complexity upper bound, but it is very loose. Alternatively, one could attempt to compute or bound the variance of this estimator and use the Chebyshev-Cantelli inequality, as done in Ref.~\cite{ohnemus2023Quantifying}, but this requires control of up to the fourth moment of a multinomial distribution, which is difficult to compute for general input states. We address these limitations with a simple shift in the estimation procedure.

\subsubsection{LRM-Mean Estimator}
 Instead of estimating the probabilities, or even the squared probabilities, we estimate the \textit{expectation value of the sum of the squared probabilities} that appear in Eq.~\eqref{eq:CE-LRM-exact-all-strings}.
\begin{theorem}[Unbiased CE Estimation via LRM Data]\label{thm:Ohnemus-LRM-Hoeffding}
For each $K\geq 2$,
    \begin{align}
    \hat{\mathcal{C}}_{\ket{\psi}}(S) = 1 - \left(\frac{3}{2}\right)^{s}\frac{1}{L} \sum_{l=1}^L \hat{S}_{l}^{(K)}, \label{eq:Ohnemus-LRM-Hoeffding}
\end{align}
is an unbiased estimator of $\mathcal{C}_{\ket{\psi}}(S)$, where
\begin{align}
    \hat{S}_l^{(K)} &= \frac{1}{K (K - 1)} \sum_{\substack{k, k' = 1\\ k \neq k'}}^{K} \ind{\bm{Z}_{l,k} = \bm{Z}_{l,k'}}. \label{eqn:sumsqprobest}
\end{align}
Given a precision $\epsilon > 0$ and a confidence level of $1 - \delta \in (0, 1)$, using at most $L = O((\frac{9}{4})^{s} \log{(1/\delta)}/\epsilon^2)$ random unitaries, one is guaranteed 
\begin{align}
    \Pr[|\hat{\mathcal{C}}_{\ket{\psi}}(S) - \mathcal{C}_{\ket{\psi}}(S)| \geq \epsilon] \leq \delta.
\end{align}
\end{theorem}
\begin{figure}
    \centering
    \includegraphics[width=0.45\textwidth]{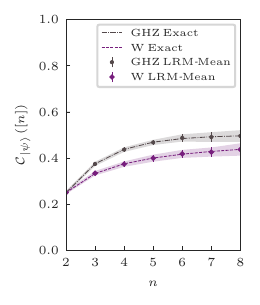}
    \caption{\textbf{CEs via LRM data.} Simulation of LRM experiment with $L = 10^4$ local Haar random unitaries and $K=2$ shots per unitary estimating the CE of the GHZ and W states on $n$ qubits. We see complete agreement with analytical results, confirming our estimator is unbiased. See Appendix~\ref{app:LRM-Simulations} for details.} 
    \label{fig:CE-LRM-GHZ-W}
\end{figure}

A detailed proof of this theorem and a derivation of the estimator can be found in Appendix~\ref{app:proof-Ohnemus-LRM-Hoeffding}. We remark that $\hat{S}_l^{(K)}$ is mathematically equivalent to summing Eq.~\eqref{eq:p-squared-unbiased} over all bitstrings, but writing it as $\hat{S}_l^{(K)}$ allows for faster classical post-processing for a constant $K$. As mentioned previously, computing each probability and summing would scale as $O(K 2^s)$, which would become a bottleneck for moderate system sizes. In contrast, computing $\hat{S}_l^{(K)}$ directly can be done in $O(K^2 s)$. For constant $K$, this represents an exponential improvement in classical post-processing. 

Importantly, the concentration inequality in Theorem~\ref{thm:Ohnemus-LRM-Hoeffding} holds for all $K \geq 2$. Thus, using a given unitary for more than $2$ shots would increase the total measurement budget without providing provably better concentration. Moreover, the computational post-processing cost increases with $K$. To see this, note that the sum $\hat{S}_l^{(K)}$ is computed by checking the number of distinct pairs of trials $(k, k')$ that output the same string, and thus, has time complexity $O(K^2 s)$. This, too, indicates that one should choose $K=2$. Note that this does not imply this is the optimal division of the total measurement budget. We will see in the next section that the sample complexity could be marginally improved by increasing $K$, but with diminishing returns after a certain value. Finally, although the classical post-processing is efficient, the sample complexity upper bound scales as $\sim (\frac{9}{4})^{s}$, which can become prohibitive for tens of qubits. As such, we now show how to quadratically improve this sample complexity using an alternative estimation procedure.

\subsubsection{LRM-MoM Estimator \label{sec:LRM-Median-of-Means}}
As mentioned in Sec.~\ref{sec:classical-stats}, Hoeffding's inequality holds for all independent random variables that are bounded. However, if the range of values the random variable takes is large relative to its standard deviation, and one can compute exactly or obtain a good upper bound on the variance of the random variable, median-of-means estimation can give tighter concentration. To apply Prop.~\ref{prop:med-of-means} to our LRM estimator, we have to compute the variance of $\hat{S}_l^{(K)}$ given in Theorem~\ref{thm:Ohnemus-LRM-Hoeffding}. As we show in Appendix~\ref{app:proof-CE-via-LRMs-MoM}, the variance can be expressed as 
\begin{align} \label{eq:var-S_l^K}
\begin{split}
    \var[\hat{S}_l^{(K)}] = \frac{2 P_2 (1 - P_2) + 4 (K - 2) (P_3 - P_2^2)}{K (K - 1)}\\
    +\frac{(K - 2) (K - 3) (P_{2, 2} - P_2^2)}{K(K-1)},
\end{split}
\end{align}
where $P_2:=\E_U[\sum_{\bm{z}}P_U(\bm{z})^2]$, $P_3:=\E_U[\sum_{\bm{z}}P_U(\bm{z})^3]$, $P_{2,2}:=\E_U[(\sum_{\bm{z}}P_U(\bm{z}))^2]$. Note, importantly, that this variance is independent of $l$ and the second additive term dominates in the large $K$ limit, approaching $P_{2,2}-P_2^2$. Thus, increasing $K$ past some threshold value will not lead to an appreciable decrease in the variance. Moreover, if one wishes to bound the variance for arbitrary states, it is necessary to bound the fourth moment implicit in $P_{2, 2}$. While local Cliffords are a common choice for implementing LRMs, they only form a $3$-design, so one would need to either find a method to bound $P_{2, 2}$ for (local) Clifford measurements, or perform (local) Haar random measurements, which can be more experimentally challenging.  We therefore restrict to $K = 2$, which allows us to use local Clifford measurements, bound the variance for arbitrary states, and provably perform better than the estimation procedure given in Theorem~\ref{thm:Ohnemus-LRM-Hoeffding}.
In addition, this allows for faster classical post-processing, as explained in the previous section.

Thus, choosing $K = 2$ (i.e. employing only two measurements per unitary), Eq.~\eqref{eq:var-S_l^K} reduces to
\begin{align}\label{eq:var-S_l^2}
    \var[\hat{S}_{l}^{(2)}] = P_2 (1-P_2).
\end{align}
Note that Eq.~\eqref{eq:CE-LRM-exact-all-strings} and the definition of $P_2$ above imply that the CE can be expressed as $\CC_{\ket{\psi}} (S) = 1-\left(\frac{3}{2}\right)^{s} P_2$.
Then, because $\CC_{\ket{\psi}} (S) \geq 0$, we have $P_2 \leq \left(\frac{2}{3}\right)^{s} $, which coupled with Eq.~\eqref{eq:var-S_l^2}, implies that  $\var[\hat{S}_{l}^{(2)}] \leq \left(\frac{2}{3}\right)^{s}$.
Noting that we can write the estimator in Eq.~\eqref{eq:Ohnemus-LRM-Hoeffding} as $\hat{\mathcal{C}}_{\ket{\psi}}(S) = \sum_{i = 1}^L \hat{\mathcal{C}}_l^{(K)} / L$ with $\hat{\mathcal{C}}_l^{(K)} = 1 - (3/2)^{s} \hat{S}_l^{(K)}$, we obtain
\begin{align}
     \var[\hat{C}_{l}^{(2)}] = \left(\frac{9}{4}\right)^{s} \var[\hat{S}_{l}^{(2)}] \leq \left(\frac{3}{2}\right)^{s}. \label{eqn:CE_LRM_K2_var_bound}
\end{align}

This allows us to directly apply Prop.~\ref{prop:med-of-means}, which we state as our second theorem.
\begin{theorem}[CE Estimation via MoM and LRMs]\label{thm:CE-via-LRMs-MoM}
    Given precision of $\epsilon > 0$ and a confidence level $1 - \delta \in (0, 1)$, we randomly sample a total of $L=N_B B$ local unitaries of the form $U = \prod_{i \in S} U_i$,
    where $N_B=\lceil 8 \log(\frac{1}{\delta})\rceil$ and $B=\lceil 4 \left(\frac{3}{2}\right)^{s}\epsilon^{-2}\rceil$.
    Measuring $K=2$ outcomes per unitary, we denote the outcomes from the $L$ LRMs as $\bm{Z}_{1,1},\bm{Z}_{1,2},\ldots,\bm{Z}_{L,1},\bm{Z}_{L,2}$.
    Breaking these $2 N_B B$ outcomes into $N_B$ batches of size $2 B$, for each $1 \leq b \leq N_B$, compute the empirical mean
    \begin{align}
        \overline{\CC}_{\ket{\psi}}^{(b)}(S)=1-\left(\frac{3}{2}\right)^{s}\frac{1}{B}\sum_{l=(b-1)B+1}^{bB}\ind{\bm{Z}_{l,1}=\bm{Z}_{l,2}}. \label{eqn:LRM_batch_mean}
    \end{align}
    Then, given at most $O\left(\left(\frac{3}{2}\right)^{s} \log\left(\frac{1}{\delta}\right) \epsilon^{-2}\right)$ measurement outcomes, one is guaranteed to have
    \begin{align}
        \Pr[|\med[\overline{\CC}^{(1)}_{\ket{\psi}}(S),\ldots,\overline{\CC}^{(N_B)}_{\ket{\psi}}(S)]- \mathcal{C}_{\ket{\psi}}(S)| \geq \epsilon] \leq \delta.
    \end{align}
\end{theorem}
The proof of this theorem can be found in Appendix~\ref{app:proof-CE-via-LRMs-MoM}, and is a direct application of Proposition~\ref{prop:med-of-means},
using the variance bound in Eq.~\eqref{eqn:CE_LRM_K2_var_bound}. While the classical post-processing needed to compute this estimator is no worse than the estimator in Theorem~\ref{thm:Ohnemus-LRM-Hoeffding}, it improves on the sample complexity by a factor of $\left(\frac{3}{2}\right)^{s}$, which effectively doubles the size of the state one could probe with comparable resources (asymptotically).

Despite this significant improvement in sample complexity and classical post-processing, LRMs still require an exponential number of measurement settings which is time consuming experimentally. To address this limitation, we now show how to de-randomize these protocols and probe entanglement using a single experimental setting.

\subsection{CEs via Projective 2-designs}
As mentioned previously, the sample complexity of an estimation protocol is not the only relevant experimental consideration. While each measurement setting used in LRMs is easily implementable, we saw above that guaranteeing $\epsilon$-close estimation of the CEs requires an exponential number of measurement settings. Inspired by Ref.~\cite{stricker2022Experimental}, we propose a ``de-randomization'' of the LRM protocol that requires only a single experimental setting. That such a protocol is possible follows from Eq.~\eqref{eq:CE-local-symmetric-subspaces} and the definition of a \textit{projective 2-design}~\cite{scott2006Tight} as stated below.

\begin{definition}[Projective 2-design]\label{def:Projective-2-designs}
    A projective 2-design is a probability distribution over $N$ quantum states, $\{p_i,\ket{\phi_i}\}_{i=1}^N$, such that
    \begin{align}
        \sum_i p_i (\ket{\phi_i}\bra{\phi_i})^{\otimes 2}
        &=\int_{\mathrm{Haar}} (\ket{\psi}\bra{\psi})^{\otimes 2} d\psi
    \end{align}
    where integration in the right-hand expression is with respect to the Haar measure.
\end{definition}
From this definition and Schur's lemma, we have
\begin{align} \label{eq:2-designs-equate-symmetric}
    \mathbb{E}_{\phi}[(\ketbra{\phi}{\phi})^{\otimes 2}]&= \frac{1}{3}\Pi_+,
\end{align}
where $\mathbb{E}_{\phi}[(\ketbra{\phi}{\phi})^{\otimes 2}]:=\sum_i p_i (\ket{\phi_i}\bra{\phi_i})^{\otimes 2}$.
Note that the state in $\mathbb{E}_{\phi}[(\ketbra{\phi}{\phi})^{\otimes 2}]$ denotes a random variable while $\ketbra{\phi_i}{\phi_i}$ is a specific state.
Thus, we may estimate all of the CEs by appropriately post-processing the measurement data from any POVM that also forms a projective 2-design. This motivates the following theorem.

\begin{theorem}[CEs via Projective 2-designs]\label{thm:CE-via-designs-exact}
    Let $\{p_i,\ket{\phi_i}\}_{i=1}^N$ be a single-qubit projective 2-design with $N$ elements. Further, let $\ketbra{\Phi_\mathbf{q}}{\Phi_\mathbf{q}}=\prod_{i\in \mathbf{q}} \ketbra{\phi_i}{\phi_i}$ denote the projector onto output string $\mathbf{q}\in\{1,2,\ldots,N \}^{s}$ which occurs with probability $P(\mathbf{q})=\tr{\rho \ketbra{\Phi_\mathbf{q}}{\Phi_\mathbf{q}}}$ when measuring the qubits in the set $S$, given the state $\rho$. Then, the CEs can be written as
    \begin{align}
        \CC_{\ket{\psi}}(S)=1-3^{s} \mathbb{E}_{\Phi}[\tr{\rho\ketbra{\Phi}{\Phi}}^2],
    \end{align}
    where $\mathbb{E}_{\Phi}[\cdot]$ denotes an expectation over the 2-design.
\end{theorem}

The proof is a direct application of Eq.~\eqref{eq:2-designs-equate-symmetric} and Eq.~\eqref{eq:CE-local-symmetric-subspaces}, keeping in mind that each $\ketbra{\phi}{\phi}$ is independent. The details are provided in Appendix~\ref{app:proof-CE-via-designs-exact}.

While Thm.~\ref{thm:CE-via-designs-exact} holds for all projective 2-designs, we will restrict our attention to SIC POVMs, because they are minimal among single-qubit projective $2$-designs. That is, they saturate the lower bound on the number of elements, $N$, needed to form a projective $t$-design in a $d$-dimensional space~\cite{scott2006Tight} 
\begin{align}
    N \geq \binom{d+\lceil t/2 \rceil -1}{\lceil t/2 \rceil}\binom{d+\lfloor t/2 \rfloor -1}{\lfloor t/2 \rfloor},
\end{align}
which becomes $N \geq 4$ for $2$-designs for the $2$-dimensional case on which we focus. Their minimal nature has important bearing on experimental implementations. Because the size of ancillary space required for implementation through Neumark's (Naimark's) Theorem (see Supp. Theorem~\ref{sec:Neumark}) is proportional to the number of design elements~\cite{chen2007ancilla}, SIC POVMs are the cheapest to implement experimentally. We refer the reader to Refs.~\cite{renes2004symmetric,scott2006Tight,fuchs2017sic} for more about SICs generally, and to Ref.~\cite{scott2006Tight} for a detailed overview of SICs and projective 2-designs as we use them. 

\begin{figure}
    \centering
    \includegraphics[width=0.45\textwidth]{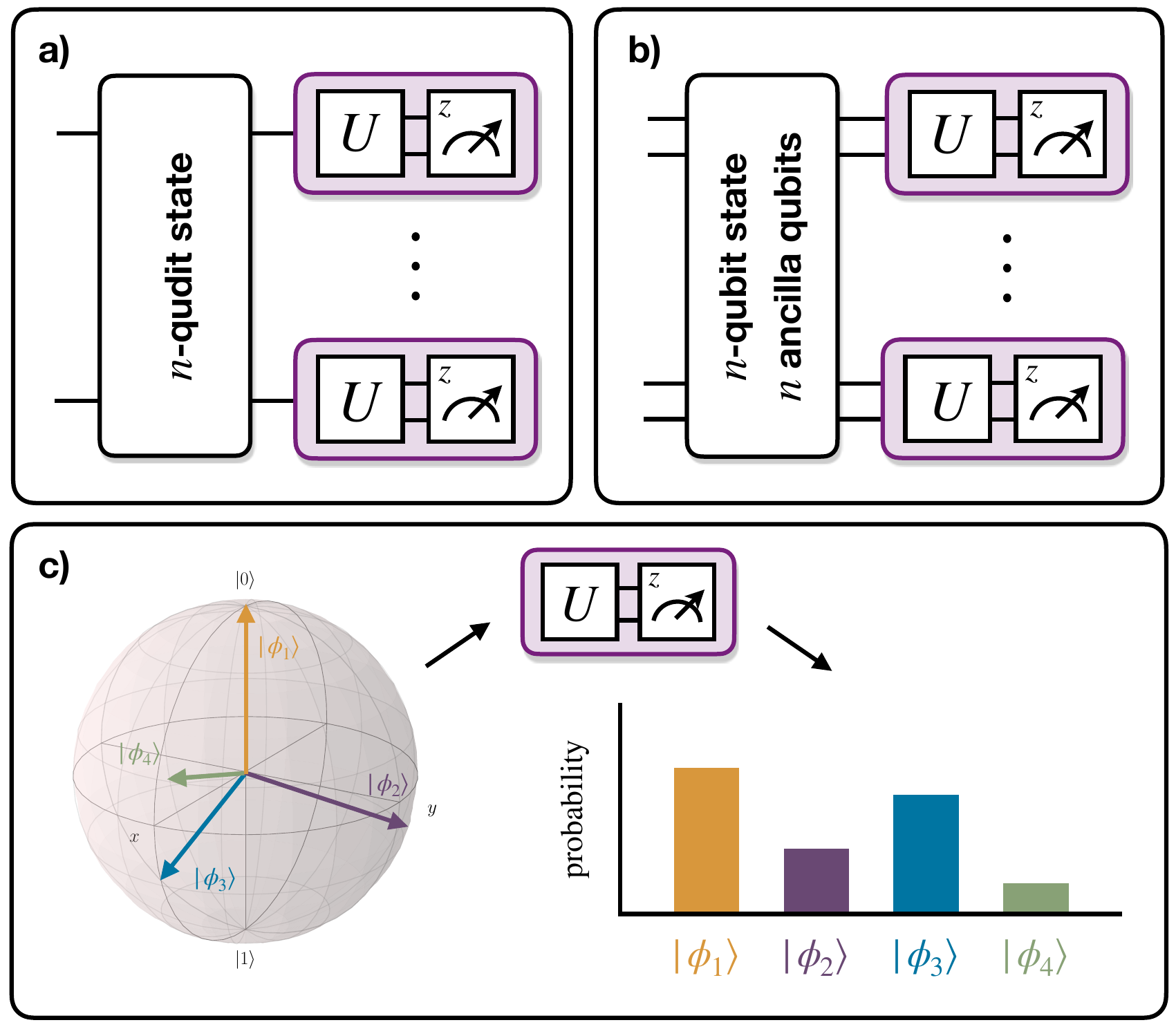}
    \caption{\textbf{Local SIC-POVM Implementation.} a) Circuit diagram showing local SIC-POVM implementation by encoding an $n$-qubit state in an $n$-qudit state. This has been implemented recently in both superconducting transmon \cite{fischer2022ancilla-free} and ion trap quantum systems \cite{stricker2022Experimental}. b) Instead of encoding qubit states in qudit states, one could utilize ancillary qubits to implement the SIC-POVM. This may be preferable in neutral atom systems as in Ref.~\cite{bluvstein2022quantum}. c) Bloch sphere representation of a single-qubit SIC-POVM. }
    \label{fig:SIC}
\end{figure}
For qubits, and as depicted in Fig.~\ref{fig:SIC}c), the simplest SIC has the following 4 elements
\begin{align}
\begin{split}
    \ket{\phi_1}&=\ket{0},\\
    \ket{\phi_2}&=\frac{1}{\sqrt{3}}\ket{0}+\sqrt{\frac{2}{3}}\ket{1},\\
    \ket{\phi_3}&=\frac{1}{\sqrt{3}}\ket{0}+\sqrt{\frac{2}{3}}e^{i2\pi/3}\ket{1},\\
    \ket{\phi_4}&=\frac{1}{\sqrt{3}}\ket{0}+\sqrt{\frac{2}{3}}e^{i4\pi/3}\ket{1},
\end{split}
\end{align}
which obey the following relation,
\begin{align}
    \frac{1}{4}\sum_{i=1}^4 (\ketbra{\phi_i}{\phi_i})^{\otimes 2} = \frac{1}{3}\Pi_+.
\end{align}
By defining $\ket{\Tilde{\phi}_i} = \frac{1}{\sqrt{2}}\ket{\phi_i}$, we can turn these SIC elements into a well-defined POVM satisfying
\begin{align}
    \sum_{i=1}^4 \ketbra{\Tilde{\phi}_i}{\Tilde{\phi}_i} &= \mathbb{I}.
\end{align}
With this notation in place, we can state a corollary to Theorem~\ref{thm:CE-via-designs-exact} that yields a very simple expression for the CEs via local SICs. 

\begin{corollary}[CEs via local SICs]\label{thm:CEs-via-SICs}
    Let $\{\ketbra{\Tilde{\phi}_i}{\Tilde{\phi}_i}\}_{i=1}^4$ be a single-qubit SIC-POVM and $\ketbra{\Tilde{\Phi}_{\mathbf{q}}}{\Tilde{\Phi}_{\mathbf{q}}}=\prod_{i \in \bm{q}} \ketbra{\Tilde{\phi}_i}{\Tilde{\phi}_i}$ denote the projector onto output string $\mathbf{q}\in\{1,2,3,4\}^{s}$ which occurs with probability,
    \begin{align}\label{eq:SIC-outcome-prob-via-POVM}
        P(\mathbf{q})=\tr{\rho \ketbra{\Tilde{\Phi}_{\mathbf{q}}}{\Tilde{\Phi}_{\mathbf{q}}}}.
    \end{align}
    Then the CEs can be written as,
    \begin{align}\label{eq:CE-via-SICs}
        \CC_{\ket{\psi}}(S)=1-3^{s}\sum_{\mathbf{q}}P(\mathbf{q})^2.
    \end{align}
\end{corollary}

Because we do not have to average over unitaries as in LRMs, the finite sampling error incurred when estimating Eq.~\eqref{eq:CE-via-SICs} will only be due to approximating the probability of obtaining a particular bitstring. This important distinction between SICs and LRMs means we just need to estimate the sum of the squared probabilities obtained from a local SIC measurements. While this theorem implies one can quantify multipartite entanglement with a single experimental measurement setting, this simplicity comes at the cost of implementing a generalized POVM. 

As shown schematically in Fig.~\ref{fig:SIC}, this can be achieved by employing one additional ancilla qubit for each system qubit~\cite{chen2007ancilla} or, as has been demonstrated experimentally in Refs.~\cite{stricker2022Experimental,fischer2022ancilla-free}, by encoding qubit states into 4-dimensional qudits (ququarts). Experimentally, Cor.~\ref{thm:CE-via-designs-exact} amounts to transforming each state-ancilla qubit pair or ququart by the unitary $U_{\mathrm{SIC}}^{\dagger}$, where
\begin{align}
    U_{\mathrm{SIC}} = 
        \begin{pmatrix}
            \frac{1}{\sqrt{2}} & \frac{1}{\sqrt{6}} & \frac{1}{\sqrt{6}} & \frac{1}{\sqrt{6}}\\
            0 & \frac{1}{\sqrt{3}} & \frac{e^{i2\pi/3}}{\sqrt{3}} & \frac{e^{i4\pi/3}}{\sqrt{3}}\\
            0 & \frac{1}{\sqrt{3}} & \frac{e^{-i2\pi/3}}{\sqrt{3}} & \frac{e^{-i4\pi/3}}{\sqrt{3}}\\
            \frac{1}{\sqrt{2}} & \frac{-1}{\sqrt{6}} & \frac{-1}{\sqrt{6}} & \frac{-1}{\sqrt{6}}
        \end{pmatrix}, \label{eq:U_SIC}
\end{align}
and then performing a computational basis measurement. The construction of $U_{\mathrm{SIC}}$ follows from Neumark's Theorem, which we prove carefully in Appendix~\ref{sec:Neumark} for the reader's convenience. The intuition is to convert the SIC elements into a $4$-dimensional orthonormal basis by appending to each $\ket{\tilde{\phi}_i}$ a respective $\ket{\tilde{\phi}^\perp_i}$ such that $\braket{\tilde{\phi}_i}{\tilde{\phi}^\perp_i}=0$.

An estimator based on SIC POVMs follows naturally from Theorem~\ref{thm:CEs-via-SICs} and Eq.~\eqref{eq:CE-local-symmetric-subspaces}. Because we saw that MoM gives tighter sample complexity upper bounds, we focus only on MoM estimation for SICs.

\begin{theorem}[CE Estimation via SIC Data and MoM]\label{thm:CEs-via-SICs-MoM-K=2}
    Given a precision $\epsilon > 0$ and confidence level $1 - \delta \in (0, 1)$, perform a total of $M = 2 N_B B$ SIC measurements,
    where $N_B=\lceil 8\log(1/\delta)\rceil$ and $B=\lceil 4 (3^{s}/\epsilon^2) \rceil$.
    Let $\mathbf{Q}_1, \ldots, \mathbf{Q}_M$ denote the outcomes obtained from these measurements.
    Break these $M$ outcomes into $N_B$ batches of size $2B$, and for each $1 \leq b \leq N_B$, compute the empirical mean
    \begin{align}
        \overline{\CC}^{(b)}_{\ket{\psi}}(S) = 1 - 3^{s} \frac{1}{B}\sum_{i=(b-1)B + 1}^{b B}\ind{\bm{Q}_{2i-1}=\bm{Q}_{2i}}.
    \end{align}
    Then, we have the guarantee that
    \begin{align}
        \Pr[|\med[\overline{\CC}^{(1)}_{\ket{\psi}}(S),\ldots,\overline{\CC}^{(N_B)}_{\ket{\psi}}(S)]- \mathcal{C}_{\ket{\psi}}(S)| \geq \epsilon] \leq \delta,
    \end{align}
    giving an upper bound of $O\left(3^{s} \log\left(\frac{1}{\delta}\right) \epsilon^{-2}\right)$ on the sample complexity of estimating $\mathcal{C}_{\ket{\psi}}(S)$ using SIC measurements.
\end{theorem}
The proof of this theorem can be found in Appendix~\ref{app:proof-CE-via-LRMs-MoM}. Note that we have restricted our attention to the $K=2$ case in the main text. Evidently, this sample complexity is substantially worse than the LRMs. We can improve upon this by optimizing the number of samples used to construct a single estimate of the CE using SIC measurements  (i.e. by optimizing $K$). As we show in Appendix~\ref{app:proof-CE-via-SICs-K_opt}, the optimal value of $K$ is given as
    \begin{align}\label{eq:K-opt}
    \begin{split}
        K_{\mathrm{opt}} &= \Biggl \lceil \frac{1}{2} \left(\frac{16}{\epsilon^2} \left(\frac{3}{2}\right)^{s} + 1\right) \\
        &+\frac{1}{2} \sqrt{\left(\frac{16}{\epsilon^2} \left(\frac{3}{2}\right)^{s} - 1\right)^2 + \frac{32}{\epsilon^2} 3^{s}} \Biggr \rceil,
    \end{split}
    \end{align}
which corresponds to the value of $K$ that minimizes the bound on the variance of the estimator for CE determined by SIC measurements.
As shown in Appendix~\ref{app:proof-CE-via-SICs-K_opt}, setting $K=K_{\mathrm{opt}}$ and appropriately modifying Theorem~\ref{thm:CEs-via-SICs-MoM-K=2}, it follows that $O\left( \left[\left(\frac{3}{2}\right)^{s} \epsilon^{-2} + \sqrt{3^{s}} \epsilon^{-1}\right] \log\left(\frac{1}{\delta}\right) \right)$ total measurements are sufficient to guarantee that
    \begin{align}
\Pr[|\med[\overline{\CC}^{(1)}_{\ket{\psi}}(s),\ldots,\overline{\CC}^{(N_B)}_{\ket{\psi}}(s)]- \mathcal{C}_{\ket{\psi}}(s)| \geq \epsilon] \leq \delta.
    \end{align}
With this optimization, we have shown that using a single experimental measurement setting, one can obtain an upper bound on the sample complexity that is similar to the bound obtained for LRMs, which requires an exponential number of different measurement settings. However, as noted in Table~\ref{table:complexities}, the classical post-processing in this optimized case is more costly. Moreover, the local measurement itself is more difficult to implement than in the case of the simple projective measurements used in LRM protocols. 

\begin{figure}[ht!]
    \centering
    \includegraphics[width=0.45\textwidth]{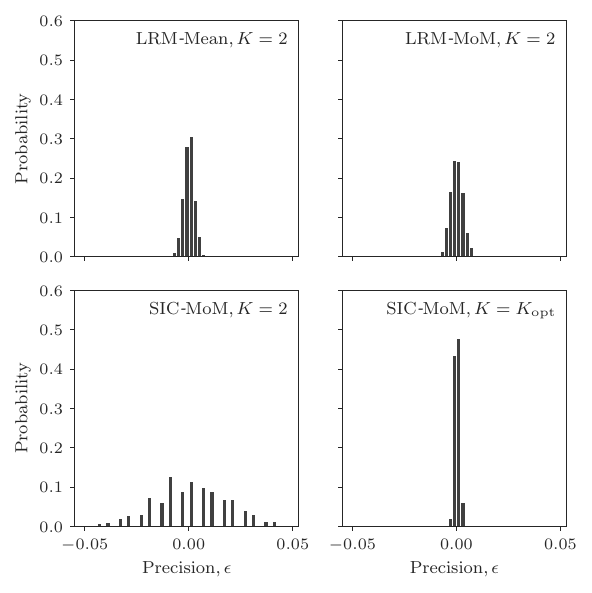}
    \caption{\textbf{Comparison of estimators at a fixed total measurement budget and confidence level}. Histograms generated from numerical simulation of all four estimators (see Table~\ref{table:complexities}) of the CE of 5-qubit GHZ state. Random local Clifford measurements were employed and the total measurement budget was computed from the upper bound on $\mathrm{SIC}\textit{-}\mathrm{MoM},K_{\mathrm{opt}}$ in the case of $\epsilon=\delta=0.05$.}
    \label{fig:4histComp-GHZ5}
\end{figure}

Nonetheless, we hope that our results will be useful to experimentalists that would like to quantify multipartite entanglement without the need to prepare and coherently manipulate multiple identical copies of a quantum state. While this capability is becoming possible on neutral atom platforms \cite{bluvstein2022quantum,bluvstein2023Logical}, it remains very challenging on other architectures.

Our motivation in this work was to improve upon protocols in the literature and find the most resource efficient method for multipartite entanglement quantification using only local measurements. While we have accomplished this, we also aimed to provide tools that experimentalists can use in current and near-term experiments where states are not perfectly pure. While there are a number of ways one could quantify mixed state entanglement, we sketch one technique to which the methods above could be directly applied.

\subsection{Mixed States}
Our motivation in this work was to optimize and compare various methods of quantifying pure state multipartite entanglement using only local POVMs. As such, we focused on pure states with ideal unitary evolution and perfect measurements. While this question is of theoretical interest, to be useful in practice, one must be able to handle mixed states. While we save the full treatment of mixed states for a future work (more on this below), we mention here one potential way forward that would utilize our results here. 

A standard method of extending pure state entanglement measures to mixed state ones is via a convex roof extension \cite{bennett1996Mixedstate,uhlmann2010Roofs}. This was done for the CEs in Ref.~\cite{beckey2023Multipartite}, but we sketch the argument here to illustrate the main points. The convex roof extension can be expressed as
\begin{align}\label{eq:mixed-CE}
    \mathcal{C}_{\rho}(S) &= \inf \sum_i p_i \mathcal{C}_{\ket{\psi_i}}(S),
\end{align}
where the infimum is over the set of decompositions of the form $\rho=\sum_i p_i \ket{\psi_i}\bra{\psi_i}$, with $p_i \geq 0$ for all $i$ and $\sum_i p_i =1$. Because this optimization is generally difficult, one often considers a lower bound on Eq.~\eqref{eq:mixed-CE}. We will show how one could construct such lower bounds for the CEs and compute them from local measurement data alone, allowing one to bound the mixed state entanglement within the above framework developed for estimating pure state entanglement. 

First we note that the bipartite concurrence \cite{wootters2001entanglement, rungta2001Universal} of a pure quantum state $\ket{\psi}_{AB}$ can be expressed as 
\begin{align}
    c_2(\ket{\psi}_{AB}) &= \sqrt{2(1-\tr{\rho_A^2})}.
\end{align}
We can use this to express the CE in terms of the bipartite concurrences by defining $c_{\alpha}:= \sqrt{2(1-\tr{\rho_{\alpha^2}})}$, leading to 
\begin{align}
    \mathcal{C}_{\ket{\psi}} (S) = \frac{1}{2^{s+1}} \sum_{\alpha} c_{\alpha}^2(\ket{\psi}).
\end{align}
Then, one could use the observable lower bound on the bipartite concurrence introduced in Ref.~\cite{mintert2007observable} to construct a lower bound on the CEs. For example, when the $S=[n]$ and the CE is simply related to the multipartite concurrence of Ref.~\cite{carvalho2004decoherence}, one obtains a lower bound of the form
\begin{align}
    \CC_{\rho}^{\ell}([n]) &= \frac{1}{2^n} + (1-\frac{1}{2^n})\tr{\rho^2} - \frac{1}{2^n}\sum_{\alpha \in \PC ([n])} \tr{\rho_{\alpha}^2}.
\end{align}
Throughout this work, we have shown how to estimate the last term using various local measurement strategies. The middle term depends only on the quantum state's purity, $\tr{\rho^2}$. There exists a number of theoretical \cite{vanenk2012measuring,elben2019statistical} and experimental \cite{daley2012measuring,islam2015measuring,kaufman2016quantum,vermersch2018unitary,brydges2019probing} works showing how to estimate the purity of quantum states using local random measurements, and local SICs were implemented experimentally in Ref.~\cite{stricker2022Experimental} to estimate purity. In fact, the purity of a quantum state is a quantity of great practical interest in its own right, and we suspect our results will be of interest to those researchers interested in proving rigorous performance guarantees on its estimation.

Coupled with these results from the literature, one could use the framework and techniques described above to probe mixed state entanglement in this manner. We further note that, for high-purity states that are becoming increasingly common in today's state of the art experiments, this bound is very close to the pure state theoretical value. This can be seen by noting that $C_{\ket{\psi}}([n]) - \CC_{\rho}^{\ell}([n]) = (1-2^{-n})(1-\tr{\rho^2}),$ which is very close to zero for nearly pure states. 

Such a procedure described above would generalize the method used in Ref.~\cite{ohnemus2023Quantifying} for the multipartite concurrence~\cite{mintert2005concurrence,mintert2007observable}, but is just one possible extension of our methods to mixed state entanglement. There have been a number of interesting works in the recent years that use local random measurement strategies to probe mixed state entanglement \cite{elben2020MixedState,neven2021Symmetryresolved,vermersch2023Manybody}, and exploring the connections between our proposed method and theirs is an interesting direction we save for future work.

\section{Future Directions and Conclusions}
Before concluding, we would like to mention some directions for future investigations and place our work in the broader context of the literature. As mentioned above, a timely follow-up to our work would be the careful analysis applying our methods to mixed states and, ideally, an implementation on real hardware.

In addition to matters of direct practical interest, our work also connects to interesting open theoretical questions regarding the ultimate limits on the learnability of quantum entanglement with local measurements. While we provide several upper bounds on the sample complexity of the estimation of multipartite entanglement with local measurements, the best of which scales as $(\frac{3}{2})^{s}$, it remains an open question as to what the optimal scaling is. Finding matching lower and upper bounds on the sample complexity of this task would be an exciting result both for experimentalists wishing to probe entanglement in the lab and to the quantum learning theory community that aims to quantify the ultimate limits of the learnability of quantum properties. In particular, such bounds would allow one to establish ultimate separations between local, single-copy, and multi-copy measurement strategies. While such separations have been established for QST (see Table 1 in Ref.~\cite{lowe2022Lower}) and other learning tasks, the authors are unaware of any such separation in the case of multipartite entanglement estimation.

However, partial results in that direction do exist. Given the ability to create, store, and coherently manipulate two copies of state at once (as in Refs.~\cite{bluvstein2022quantum,bluvstein2023Logical}), one could use Bell basis measurements on each qubit of the two copies of $\rho$ to estimate the CEs~\cite{beckey2023Multipartite}. Given access to these two-copy measurements, Ref.~\cite{beckey2023Multipartite} showed that at most $O(\log{(\frac{1}{\delta})} \epsilon^{-2})$ measurements are needed to estimate the CEs to $\epsilon$ precision with probability $1-\delta$. Thus, for constant $\epsilon,\delta$, two-copy measurements allow for multipartite entanglement quantification using at most $O(1)$ measurements, while our best local strategy scales as $O((\frac{3}{2})^s)$. Proving an exponential lower bound on multipartite entanglement estimation with local measurements would imply an exponential separation between local and two-copy measurements, which would be very interesting to theorists and experimentalists alike.

To conclude, in this work we have generalized Ref.~\cite{ohnemus2023Quantifying}, which studies the estimation of multipartite concurrence using LRMs, to the CEs and then significantly simplified the error analysis and provided analytical upper bounds on the sample complexity of this task, as summarized in Table~\ref{table:complexities}. While each measurement in an LRM protocol is easy to implement, the number of measurement \textit{settings} required scales exponentially with the number of qubits. To address this experimental shortcoming of LRMs, we provided a \textit{de-randomization} of the entanglement estimation procedure, allowing experimentalists to estimate many multipartite entanglement measures of interest using a \textit{single experimental setting}.

Finally, the purity of a quantum state is a quantity of great practical interest, in its own right. Many works exploring the estimation of functions of the form $\tr{\rho^k}$ for integer $k$ fall short of the rigorous performance guarantees we provide here. We suspect our methods of statistical estimation could be adapted to those settings as well. In general, we hope our work will enable multipartite entanglement quantification in the increasingly large quantum systems being built, and coherently controlled, in experimental laboratories today.

\begin{acknowledgments}
The authors thank Chris Fuchs for several enlightening discussions about SIC measurements during his visit to JILA. This work was supported by NSF award 2137984 and NSF PHY 1915407.
\end{acknowledgments}

\bibliography{main.bib}

\setcounter{section}{0}
\setcounter{proposition}{3}
\setcounter{theorem}{4}
\setcounter{lemma}{1}
\setcounter{corollary}{1}
\setcounter{figure}{0}
\renewcommand{\figurename}{Sup. Fig.}

\appendix
\section*{Appendix}
\maketitle

\section{Preliminaries}
\subsection{Neumark's (Naimark's) Theorem}\label{sec:Neumark}
Neumark's theorem is a procedure for realizing POVMs as a projective measurement on a larger Hilbert space. Because it is the key theoretical component that allows SICs to be implemented in practice, we provide a proof here. Our proof is based directly upon the one given in Preskill's 1997 lecture notes \cite{preskillNotes}. Note that we focus on the case of POVMs with four elements due to our interest in single qubit SICs; however, the result can be generalized for POVMs with $N$ elements.

\begin{theorem}[Neumark's (Naimark's) Theorem for single-qubit SICs]\label{thm:Neumark's}
    Suppose our two-dimensional Hilbert space of interest $\mathcal{H}_A$ is actually a subspace of a larger Hilbert space with a direct sum structure $\HC = \HC_A \oplus \HC_A^{\perp}$, where $\HC_A^\perp$ is another two-dimensional Hilbert space. Then, a single-qubit SIC-POVM has elements $\{ F_i\}_{i=0}^3=\{ \frac{1}{2} \ketbra{\phi_i}{\phi_i}\}_{i=0}^3 = \{ \ketbra{\Tilde{\phi_i}}{\Tilde{\phi_i}}\}_{i=0}^3\subset \HC_A$ that only have support on $\HC_A$, i.e, $F_i \ket{\phi^\perp} = 0 = \bra{\phi^\perp}F_i$,
    for any $\ket{\phi^\perp}\in \HC_A^\perp$ and for any~$i$.
    We can then realize the SIC-POVM as a projective measurement $\{\ketbra{u_i}{u_i} \}$ on $\HC$, where $\{\ket{u_0}, \dotsc, \ket{u_3}\}$ is an orthonormal basis, with $\ket{u_i}$ defined as
    \begin{align}
        \ket{u_i} = \ket{\Tilde{\phi_i}} \oplus \ket{\Tilde{\phi}_i^\perp},
    \end{align}
    where $\ket{\Tilde{\phi}_i^\perp}$ is a vector of magnitude $1/2$ that is orthogonal to $\ket{\Tilde{\phi}_i}$ for each $i$ given in Eq.~\eqref{eqn:psi_perp_SIC}.
\end{theorem}
\begin{proof}
    We have the following requirements to realize the single-qubit SIC-POVM as a projective POVM on the larger space $\mathcal{H}$:
    \begin{enumerate}
        \item $\{\ket{u_0}, \dotsc, \ket{u_3}\}$ forms an orthonormal basis for $\mathcal{H}$, and
        \item For any state $\rho$ on the space $\mathcal{H}_A$, we have
              \begin{align}
                  \tr{(\rho \oplus 0_{A^\perp})\ketbra{u_i}{u_i}} = \tr{\rho \ketbra{\Tilde{\phi_i}}{\Tilde{\phi_i}}}
              \end{align}
              for all $i \in \{0, \dotsc, 3\}$, where $0_{A^\perp}$ denotes the zero matrix on $\mathcal{H}_A^\perp$.
    \end{enumerate}
    The second requirement above ensures that the POVM on the larger space reproduces the statistics of the single-qubit SIC POVM.
    This condition implies that each $\ket{u_i}$ is of the form $\ket{u_i} = \ket{\Tilde{\psi}_i} \oplus \ket{w_i}$ for some vector $\ket{w_i} \in \mathcal{H}_A^\perp$.
    Then, orthonormality of the $\{\ket{u_i}\}$ basis gives the constraint that $\ip{w_j}{w_i} = -\ip{\Tilde{\phi}_j}{\Tilde{\phi}_i}$ for all $i, j$.
    It can be verified through explicit calculation that choosing $\ket{w_i} = \ket{\Tilde{\phi}_i^\perp}$ to be a vector perpendicular to $\ket{\Tilde{\phi}_i}$ for each $i$ given in the equation below satisfies the above requirements.
    \begin{align}
    \begin{split}
        \ket{\Tilde{\phi}_0^\perp}&=\frac{1}{\sqrt{2}}\ket{1},\\
        \ket{\Tilde{\phi}_1^\perp}&=\frac{1}{\sqrt{3}}\ket{0}-\frac{1}{\sqrt{6}}\ket{1},\\
        \ket{\Tilde{\phi}_2^\perp}&=\frac{e^{-i 2\pi/3}}{\sqrt{3}}\ket{0}-\frac{1}{\sqrt{6}}\ket{1},\\
        \ket{\Tilde{\phi}_3^\perp}&=\frac{e^{-i 4\pi/3}}{\sqrt{3}}\ket{0}-\frac{1}{\sqrt{6}}\ket{1}.
    \end{split}
    \label{eqn:psi_perp_SIC}
    \end{align}
\end{proof}
Neumark's Theorem involves a direct sum structure, but it can be exchanged for a tensor product structure with some caveats \cite{chen2007ancilla}.
For the case of single-qubit SIC-POVM, however, this is not a problem as 
both $\mathcal{H}_A$ and $\mathcal{H}_A^\perp$ are two-dimensional.
Specifically, we can obtain a tensor product structure by considering an ancillary qubit $B$, and writing
\begin{align}
   \ket{u_i} = \ket{\Tilde{\phi_i}}_A\ket{0}_B + \ket{\Tilde{\phi_i^\perp}}_A\ket{1}_B.
\end{align}
This once again gives the same measurement statistics as the POVM,
\begin{align}
    \tr{\rho \ketbra{\Tilde{\phi_i}}{\Tilde{\phi_i}}} = \tr{ (\rho \otimes \ketbra{0}{0}) \ketbra{u_i}{u_i}}.
\end{align}
To physically implement a SIC-POVM we can construct a unitary to act on our state such that a computational basis measurement afterward will give us the outcomes for a SIC. Let $U_{\mathrm{SIC}} = \begin{bmatrix}
        \ket{u_1} & \ket{u_2} & \ket{u_3} & \ket{u_4}
    \end{bmatrix},$
then
\begin{align}
    \tr{\rho \ketbra{\Tilde{\phi_i}}{\Tilde{\phi_i}}} = \tr{ U_{\mathrm{SIC}}^{\dagger}(\rho \otimes \ketbra{0}{0})U_{\mathrm{SIC}} \ketbra{q}{q}}
\end{align}
where $q\in \{0,1,2,3\}$. An explicit matrix for $U_\mathrm{SIC}$ is given in Eq.~\eqref{eq:U_SIC}.

\section{Proofs of Main Results} \label{app:main-results}
\subsection{LRMs} \label{app:main-results-LRMs}
\begin{proof}[Proof of Proposition.~\ref{prop:CE-via-LRM-exact}]
    Let $S\subseteq [n]$ denote the system that we are interested in. We define $U$ as follows:
    \begin{align}
        U=\prod_{i\in S}U_{i},
    \end{align}
    where $U_i \in U(2)$. We can then compute the Haar average quantity $ \mathbb{E}[P_{U}(\mathbf{z})^2]$ as
    \begin{align}
    \begin{split}
        \mathbb{E}[P_{U}(\mathbf{z})^2] &= \mathbb{E}_{U}[\tr{\rho U \ketbra{\mathbf{z}}{\mathbf{z}} U^{\dagger}]^2},\\
        &= \mathbb{E}_{U}[\tr{\rho U \ketbra{\mathbf{z}}{\mathbf{z}} U^{\dagger}}\tr{\rho U \ketbra{\mathbf{z}}{\mathbf{z}} U^{\dagger}}],\\
        &=\mathbb{E}_{U}[\tr{\rho U \ketbra{\mathbf{z}}{\mathbf{z}} U^{\dagger} \otimes \rho U \ketbra{\mathbf{z}}{\mathbf{z}} U^{\dagger}}],\\
        &=\mathbb{E}_{U}[\tr{\rho^{\otimes 2} (U \ketbra{\mathbf{z}}{\mathbf{z}} U^{\dagger})^{\otimes 2}}],\\
    &=\tr{\rho^{\otimes2}\mathbb{E}_U[(U\ketbra{\mathbf{z}}{\mathbf{z}}U^{\dagger})^{\otimes 2}]},\\
        &=\tr{\rho^{\otimes 2}\prod_{i\in s}\mathbb{E}_{U_i}[U_i^{\otimes 2}\ketbra{z_i}{z_i}^{\otimes 2}U^{\dagger^{\otimes 2}}_i]},\\
        \mathbb{E}[P_{U}(\mathbf{z})^2] &=\left(\frac{1}{3}\right)^{\lvert s \rvert}\tr{\rho^{\otimes 2}\prod_{i\in s}\Pi^i_{+}},
    \end{split}
    \end{align}
where the last line follows from:
\begin{align}
\begin{split}
    \mathbb{E}_{U}[U^{\otimes 2}\ketbra{\mathbf{z}}{\mathbf{z}}U^{\dagger^{\otimes 2}}] &= \frac{1}{\tr{\Pi_+}}\Pi_{+}\tr{\Pi_{+}\ketbra{\mathbf{z}}{\mathbf{z}}^{\otimes 2}}\\
    &+\frac{1}{\tr{\Pi_{-}}}\Pi_{-}\tr{\Pi_{-}\ketbra{\mathbf{z}}{\mathbf{z}}^{\otimes 2}},\\
    &=\frac{2}{d(d+1)}\Pi_{+}\tr{\frac{\mathbb{I}+\mathbb{F}}{2}\ketbra{\mathbf{z}}{\mathbf{z}}^{\otimes 2}},\\
    &=\frac{2}{d(d+1)}\Pi_{+}.
\end{split}
\end{align}
Where we notice that since $\ketbra{\mathbf{z}}{\mathbf{z}}$ is a product state, under the SWAP operation ($\mathbb{F}$) it remains the same and is thus annihilated by the anti-symmetric projector. Using the relation that $\Pi_{+}=\frac{\mathbb{I}+\mathbb{F}}{2}$, $d=2$, and Lemma \ref{lemma:swap-trick}, i.e. the SWAP trick, we can rewrite to the following:
\begin{align}
    \mathbb{E}[P_{U}(\mathbf{z})^2] &= \left(\frac{1}{3}\right)^{\lvert s \rvert}\tr{\rho^{\otimes 2}\prod_{i\in s}\Pi_+^i},\\
    &= \left(\frac{1}{6}\right)^{\lvert s \rvert}\tr{\prod_{i\in S}(\mathbb{F}_{i}+\mathbb{I}_i)\rho^{\otimes 2}},\\
    &=\left(\frac{1}{6}\right)^{\lvert s \rvert}\sum_{\alpha \in \mathcal{P}(s)}\tr{\rho_{\alpha}^2},\\
    \implies &\mathcal{C}_{\ket{\psi}}(S)=1-3^{\lvert s \rvert}\mathbb{E}[P_U(\mathbf{z})^2]
\end{align}
\end{proof}

\begin{proof}[Proof of Theorem~\ref{thm:Ohnemus-LRM-Hoeffding}]\label{app:proof-Ohnemus-LRM-Hoeffding}
    Let $\bm{Z}_{l,k}$ denote the $k$th outcome observed after rotating a given subset $S$ of qubits by $U_l$ and measuring in the computational basis of those qubits.
    For a fixed $\bm{z}\in\{0,1\}^{s}$ bitstring, let $\widehat{P}_l(\bm{z})=\frac{1}{K}\sum_{k=1}^K \ind{\bm{Z}_{l,k}=\bm{z}}$ denote the fraction of $K$ outcomes equal to $\bm{z}$.
    Observe that $\widehat{P}_l(\bm{z})$ is an unbiased estimator of $\mathbb{E}_U[P_U(\bm{z})]$, which is the probability of observing the bitstring $\bm{z}$, averaged over unitaries.
    Note that here we take the expectation value over the outcome probabilities as well as the unitaries.
    Then, following Ohnemus \textit{et al.}~\cite{ohnemus2023Quantifying},
    \begin{align} \label{eq:p-squared-unbiased}
         \widehat{P}^{(2)}_l({\bm{z}}) &= \widehat{P}_l({\bm{z}})\frac{(K \widehat{P}_l({\bm{z}}) - 1)}{K - 1}.
    \end{align}
    is an unbiased estimator of the squared probability averaged over unitaries, $\mathbb{E}_U[P_U(\bm{z})^2]$.
    Consequently, $\hat{S}_l^{(K)} = \sum_{\bm{z} \in \{0, 1\}^s} \widehat{P}^{(2)}_l(\bm{z})$ is an unbiased estimator of the sum of squared probabilities averaged over unitaries.
    To see that $\hat{S}_l^{(K)}$ coincides with Eq.~\eqref{eqn:sumsqprobest}, we substitute the expression for $\widehat{P}_l(\bm{z})$ in terms of indicator function to obtain 
    \begin{align}
        \hat{S}_l^{(K)} = \frac{1}{K (K - 1)} \sum_{\substack{k, k' = 1\\ k \neq k'}}^{K} \sum_{\bm{z}} \ind{\bm{Z}_{l,k} = \bm{z}} \ind{\bm{Z}_{l,k'} = \bm{z}}.
    \end{align}
    Observe that $\sum_{\bm{z}} \ind{\bm{Z}_{l,k} = \bm{z}} \ind{\bm{Z}_{l,k'} = \bm{z}}$ is $1$ when $\bm{Z}_{l,k} = \bm{Z}_{l,k'}$ and $0$ otherwise.
    Thus, we can write $\sum_{\bm{z}} \ind{\bm{Z}_{l,k} = \bm{z}} \ind{\bm{Z}_{l,k'} = \bm{z}} = \ind{\bm{Z}_{l,k} = \bm{Z}_{l,k'}}$, and subsequently, we obtain
    \begin{align}
        \hat{S}_l^{(K)} = \frac{1}{K (K - 1)} \sum_{\substack{k, k' = 1\\ k \neq k'}}^{K} \ind{\bm{Z}_{l,k} = \bm{Z}_{l,k'}}.
    \end{align}
    As a result, $\hat{\mathcal{C}}_l^{(K)}(S) = 1 - (3/2)^s \hat{S}_l^{(K)}$ is an unbiased estimator of CE for each $l$, and $\hat{\mathcal{C}}_{\ket{\psi}}(S)$ defined in Eq.~\eqref{eq:Ohnemus-LRM-Hoeffding} is just the empirical average of $\hat{\mathcal{C}}_l^{(K)}(S)$ over $L$ randomly sampled unitaries.
    Then, since $\hat{S}_l^{(K)}$ is bounded between $0$ and $1$, using Hoeffding's inequality (Proposition~\ref{fact:hoeffding}),
    one obtains a sample complexity of $O\left(\left(\frac{9}{4}\right)^{s}\log(\frac{1}{\delta})\epsilon^{-2}\right)$ for estimating the CE.
\end{proof}
\begin{proof}[Proof of Theorem~\ref{thm:CE-via-LRMs-MoM}]\label{app:proof-CE-via-LRMs-MoM}
    From the proof of Theorem~\ref{thm:Ohnemus-LRM-Hoeffding}, we know that the the quantity $\hat{\mathcal{C}}_l^{(2)} = 1 - (3/2)^s \hat{S}_l^{(2)}$ is an unbiased estimator of $\mathcal{C}_{\ket{\psi}}(S)$.
    From Eq.~\eqref{eqn:CE_LRM_K2_var_bound}, we know that variance of $\hat{\mathcal{C}}_l^{(2)}$ is bounded above by $(3/2)^s$.
    Then applying Propostion~\ref{prop:med-of-means} yields the advertised sample complexity of $O\left(\left(\frac{3}{2}\right)^{s} \log\left(\frac{1}{\delta}\right) \epsilon^{-2}\right)$.
    The only remaining component of the proof is calculating the variance of $\hat{\mathcal{C}}_l^{(K)}$, which is done in Lemma~\ref{lem:sumsqprobest_var}.
\end{proof}

\begin{lemma}
\label{lem:sumsqprobest_var}
The variance of the estimator $\hat{S}_l^{(K)}$ defined in Eq.~\eqref{eqn:sumsqprobest} is given by
\begin{align}
    \var[\hat{S}_l^{(K)}] &= \frac{2 P_2 (1 - P_2) + 4 (K - 2) (P_3 - P_2^2)}{K (K - 1)} \nonumber \\
                          &\qquad  + \frac{(K - 2) (K - 3) (P_{2, 2} - P_2^2)}{K(K-1)}. \label{eqn:sumsqprobest_var}
\end{align}
\end{lemma}
\begin{proof}
    For any $l \in [L]$, we define $X_{k,k'}=\ind{Z_{l,k}=Z_{l,k'}}$ to be the Bernoulli random variable with mean $P_2=\E_U[\sum_{\bm{z}}P(\bm{z})^2]$.
    Then, we can write
    \begin{equation}
        \var[\hat{S}_l^{(K)}]=\frac{1}{K^2(K-1)^2}\sum_{\substack{k,k'=1\\k\neq k'}}^K\sum_{\substack{j,j'=1\\j\neq j'}}^K \cov[X_{k,k'},X_{j,j'}], \label{eqn:sumsqprobest_var_intermediate}
    \end{equation}
    which we evaluate using combinatorial arguments.
    To facilitate counting, we compare the indices $(k, k')$ and $(j, j')$ under the requirement that $k \neq k'$ and $j \neq j'$ based on this contraint in the sum above.
    Observe that $X_{k, k'} = X_{k', k}$, and thus, we need to account for indices $(k, k')$ and $(k', k)$ denoting the same random variable.
    To proceed, we break the indices appearing in the sum into three cases and evaluate each of them separately.
    \begin{enumerate}
        \item[Case 1:] $(k, k') = (j, j')$ and its permutation $(k, k') = (j', j)$. There are $2 K(K - 1)$ such terms including the permutation. This results in $X_{k, k'} = X_{j, j'}$, and thus, $\cov(X_{k, k'} X_{j, j'}) = \var(X_{k, k'}) = P_2 (1 - P_2)$.
        \item[Case 2:] $k = j$ \& $k' \neq j'$ and its $3$ permutations ($k = j'$ \& $k' \neq j$, $k' = j$ \& $k \neq j'$, $k' = j'$ \& $k \neq j$). There are a total of $4 K (K - 1) (K - 2)$ such terms including permutations. For $k = j$ \& $k' \neq j'$, we have $X_{k, k'} X_{j, j'} = \ind{\bm{Z}_k = \bm{Z}_{k'} = \bm{Z}_{j'}}$, which takes the value $1$ when exactly $3$ independent outcome strings are equal and $0$ otherwise. Subsequently, $\cov(X_{k, k'}, X_{j, j'}) = \mathbb{E}[X_{k, k'} X_{j, j'}] - \mathbb{E}[X_{k, k'}] \mathbb{E}[X_{j, j'}] = P_3 - P_2^2$, where $P_3 = \mathbb{E}_U[\sum_{\bm{z}} P_U(\bm{z})^3]$. The permutations of $(k, k')$, $(j, j')$ noted above give the same value for the covariance.
        \item[Case 3:] The last case corresponds to the situation where none of the indices are equal. There are $K (K - 1) (K - 2) (K - 3)$ such terms. In this case, the random variables $X_{k, k'}$ and $X_{j, j'}$ are independent with respect to the outcome probabilities, and subsequently, $\cov(X_{k, k'}, X_{j, j'}) = \mathbb{E}[X_{k, k'} X_{j, j'}] - \mathbb{E}[X_{k, k'}] \mathbb{E}[X_{j, j'}] = \mathbb{E}_U[(\sum_{\bm{z}} P_U(\bm{z})^2)^2] - P_2^2$. For convenience, we denote $P_{2, 2} = \mathbb{E}_U[(\sum_{\bm{z}} P_U(\bm{z})^2)^2]$.
    \end{enumerate}
    Using the values calculated in the above cases in Eq.~\eqref{eqn:sumsqprobest_var_intermediate}, one can infer that the variance of $\hat{S}_l^{(K)}$ is given by Eq.~\eqref{eqn:sumsqprobest_var}. Importantly, $\var[\hat{S}_l^{(K)}]$ is independent of $l$.
\end{proof}

In Theorem~\ref{thm:CE-via-LRMs-MoM}, we restrict our attention to $K = 2$.
Below, we briefly mention how one can obtain bounds on the variance for any $K \geq 2$.
First, we bound the variance from above by dropping terms corresponding to $-P_2^2$ to obtain
\begin{align*}
    \var[\hat{S}_l^{(K)}] &\leq \frac{2 P_2 + 4 (K - 2) P_3}{K (K - 1)}
                                    + \frac{(K - 2) (K - 3) P_{2, 2}}{K(K-1)}.
\end{align*}
As noted in Section~\ref{sec:LRM-Median-of-Means}, we have the bound $P_2 \leq (2/3)^s$ using the expression for CE given in Eq.~\eqref{eq:CE-LRM-exact-all-strings}.
To bound $P_3$, we integrate $P_U(\bm{z})^3$ over the Haar measure, which amounts to computing the third moment.
\begin{align}
\begin{split}
    \mathbb{E}_U\left[P_U(\bm{z})^3\right]&=\tr{\rho^{\otimes 3}\prod_{i \in S}\mathbb{E}_{U_i}\left[U_i^{\dagger^{\otimes 3}}\ketbra{\bm{z}}{\bm{z}}^{\otimes 3}U_i^{\otimes 3}\right]},\\
    &=\frac{1}{4^{s}}\tr{\rho^{\otimes 3}\prod_{i \in S}\Pi^i_{\mathrm{sym},3}},\\
    &\leq \frac{1}{4^{s}},
\end{split}
\end{align}
where $\Pi^i_{\mathrm{sym},3}$ denotes the projector onto the symmetric subspace defined by the symmetric group $S_3$ for the $i$-th index~\cite{mele2023Introduction}.
Thus, we obtain $P_3 = \sum_{\bm{z}} \mathbb{E}_U[P_U(\bm{z})^3] \leq 1/2^s$.
Finally, for bounding $P_{2, 2}$, we need to compute the fourth moment.
As mentioned previously, since Cliffords only form a $3$-design, the value of $P_{2, 2}$ will differ for local Cliffords and local Haar random unitaries.
We leave this computation for future work, and instead give a simple bound on $P_{2, 2}$ here.
Since $\sum_{\bm{z}} P_U(\bm{z})^2 \leq 1$ for any unitary $U$, we obtain $P_{2, 2} \leq P_2$.
This bound can likely be tightened by directly computing $P_{2, 2}$ as noted above.
In any case, since the variance does not go to zero and instead approaches $P_{2, 2} - P_2^2$ for large $K$,
sampling and measuring many unitaries is unavoidable for estimating CE to a small enough precision.
Moreover, a larger value of $K$ leads to a larger computational cost.
For this reason, we focus on $K = 2$ in this study.
\subsection{SICs}\label{app:main-results-SICs}

\begin{proof}[Proof of Theorem.~\ref{thm:CE-via-designs-exact}]\label{app:proof-CE-via-designs-exact}
    Let $S \subseteq [n]$ denote the system of interest. Let $\{p_i,\ket{\phi_i}\}_{i=1}^N$ be a projective 2-design with associated probability distribution $\phi$ such that $\mathbb{E}_{\phi}[(\ketbra{\phi}{\phi})^{\otimes 2}]=\frac{1}{3}\Pi_+$. Finally denote the joint probability distribution of the local projective 2-designs as $\Phi$ and have $\ketbra{\Phi}{\Phi}=\prod_S \ketbra{\phi}{\phi}$ denote the tensor product of local random variable states. Then the expectation of the squared projection of $\rho$ onto $\ketbra{\Phi}{\Phi}$ w.r.t. it's probability distribution is
    \begin{align}
    \begin{split}
        \mathbb{E}_{\Phi}[\tr{\rho \ketbra{\Phi}{\Phi}}^2] &= \mathbb{E}_{\Phi}[\tr{\rho \ketbra{\Phi}{\Phi}}\tr{\rho \ketbra{\Phi}{\Phi}}],\\
        &=\mathbb{E}_{\Phi}[\tr{(\rho \ketbra{\Phi}{\Phi}) \otimes (\rho \ketbra{\Phi}{\Phi})}],\\
        &=\mathbb{E}_{\Phi}[\tr{\rho^{\otimes 2}(\ketbra{\Phi}{\Phi})^{\otimes 2}}],\\
        &=\tr{\rho^{\otimes 2}\mathbb{E}_{\Phi}[(\ketbra{\Phi}{\Phi})^{\otimes 2}]},\\
        &=\tr{\rho^{\otimes 2}\mathbb{E}_{\Phi}\left[\prod_S (\ketbra{\phi}{\phi})^{\otimes 2}\right]},\\
        &=\tr{\rho^{\otimes 2}\prod_S \mathbb{E}_{\phi}[(\ketbra{\phi}{\phi})^{\otimes 2}]},\\
        &=\left(\frac{1}{3}\right)^{s}\tr{\rho^{\otimes 2} \prod_{i\in S} \Pi_+^i},\\
        \mathbb{E}_{\Phi}[\tr{\rho \ketbra{\Phi}{\Phi}}^2]&=\left(\frac{1}{6}\right)^{\lvert s \rvert}\sum_{\alpha \in \mathcal{P}(s)}\tr{\rho_{\alpha}^2}.
    \end{split}
    \end{align}
    Here we used the independence of the marginal probability distributions to take the expectation of each random state variable in the tensor product separately, then applied Lemma \ref{lemma:swap-trick}, i.e. the SWAP trick. After shuffling around factors we arrive at,
    \begin{align}
        \CC_{\ket{\psi}}(S)=1-3^{s} \mathbb{E}_{\Phi}[\tr{\rho\ketbra{\Phi}{\Phi}}^2].
    \end{align}
    
\end{proof}

\begin{proof}[Proof of Theorem~\ref{thm:CEs-via-SICs-MoM-K=2}]\label{app:proof-CE-via-SICs-MoM-K=2}
    Given two iid samples $\bm{Q}_{2i - 1}, \bm{Q}_{2i}$, we define $\hat{S}^{(K)} = \ind{\bm{Q}_{2i - 1}, \bm{Q}_{2i}}$, which is an unbiased estimator of $P_2 := \sum_{\bm{q}} P(\bm{q})^2$.
    Thus, $\hat{\mathcal{C}}_i^{(2)} = 1 - 3^s \hat{S}^{(2)}$ is an unbiased estimator for $\mathcal{C}_{\ket{\psi}}(S)$ using SIC measurements.
    From Lemma~\ref{lem:SIC_sumsqprobest_K_var}, we know that the variance of $\hat{S}^{(2)}$ is bounded above by $P_2$ for $K = 2$, and therefore, $\var[\hat{\mathcal{C}}_i^{(2)}] \leq 9^s P_2$.
    From the expression of CE in Eq.~\eqref{eq:CE-via-SICs} in terms of SIC measurements, we know that $P_2 \leq 1/3^s$, giving $\var[\hat{\mathcal{C}}_i^{(2)}] \leq 3^s$.
    Then, using Proposition~\ref{prop:med-of-means}, we obtain the desired result.
\end{proof}

\begin{theorem}[CE Estimation via SIC Data and MoM]\label{thm:CEs-via-SICs-MoM-K_Opt}
    Given a precision $\epsilon > 0$ and confidence level $1 - \delta \in (0, 1)$, perform a total of $M = N_B K_{\mathrm{opt}}$ SIC measurements,
    where $N_B=\lceil 8\log(1/\delta)\rceil$ and
    \begin{align}
    \begin{split}
        K_{\mathrm{opt}} &= \Biggl \lceil \frac{1}{2} \left(\frac{16}{\epsilon^2} \left(\frac{3}{2}\right)^{s} + 1\right) \\
                         &\qquad + \frac{1}{2} \sqrt{\left(\frac{16}{\epsilon^2} \left(\frac{3}{2}\right)^{s} - 1\right)^2 + \frac{32}{\epsilon^2} 3^{s}} \Biggr \rceil. \label{eqn:SIC_K_opt}
    \end{split}
    \end{align}
    Denote $\mathbf{Q}_1, \ldots, \mathbf{Q}_M$ to be the outcomes obtained from these measurements.
    Break these $M$ outcomes into $N_B$ batches of size $K_{\mathrm{opt}}$, and for each $1 \leq b \leq N_B$, compute estimate
    \begin{align}
        \overline{\mathcal{C}}_{\ket{\psi}}^{(b)}(S) = 1 - 3^{s} \hat{S}_b^{(K_{\mathrm{opt}})}, \label{eqn:SIC_K_opt_CE_estimator}
    \end{align}
    where
    \begin{align}
        \hat{S}_b^{(K_{\mathrm{opt}})} &= \frac{1}{K_{\mathrm{opt}} (K_{\mathrm{opt}} - 1)} \sum_{\substack{k, k' = (b - 1) K_{\mathrm{opt}} + 1\\ k \neq k'}}^{b K_{\mathrm{opt}}} \ind{\bm{Q}_{k} = \bm{Q}_{k'}}.
    \end{align}
    Then, we have the guarantee that
    \begin{align}
        \Pr[|\med[\overline{\CC}^{(1)}_{\ket{\psi}}(S),\ldots,\overline{\CC}^{(N_B)}_{\ket{\psi}}(S)]- \mathcal{C}_{\ket{\psi}}(S)| \geq \epsilon] \leq \delta,
    \end{align}
    giving an upper bound of
    \begin{equation}
        O\left(\left[\left(\frac{3}{2}\right)^{s} \epsilon^{-2} + \sqrt{3^{s}} \epsilon^{-1}\right] \log\left(\frac{1}{\delta}\right)\right)
    \end{equation}
    on the sample complexity of estimating $\mathcal{C}_{\ket{\psi}}(S)$ using SIC measurements.
\end{theorem}
\begin{proof}\label{app:proof-CE-via-SICs-K_opt}
    From Lemma~\ref{lem:SIC_sumsqprobest_K_var}, we know that for any $K \geq 2$,
    \begin{equation*}
        \hat{S}^{(K)} = \frac{1}{K (K - 1)} \sum_{\substack{k, k' = 1\\ k \neq k'}}^K \ind{\bm{Q}_{k} = \bm{Q}_{k'}}
    \end{equation*}
    is an unbiased estimator of sum of squared probabilities for SIC measurents, with
    \begin{equation*}
        \var[\hat{S}^{(K)}] = \frac{2 P_2 (1 - P_2) + 4 (K - 2) (P_3 - P_2^2)}{K (K - 1)},
    \end{equation*}
    where $P_2 = \sum_{\bm{q}} P(\bm{q})^2$ and $P_3 = \sum_{\bm{q}} P(\bm{q})^3$.
    Dropping $-P_2^2$ terms form the above equation for variance, we obtain the bound
    \begin{equation*}
        \var[\hat{S}^{(K)}] \leq \frac{2 P_2 + 4 (K - 2) P_3}{K (K - 1)}.
    \end{equation*}
    To proceed, we need to bound $P_2$ and $P_3$. As mentioned in the proof of Theorem~\ref{thm:CEs-via-SICs-MoM-K=2}, we have $P_2 \leq 1/3^s$.
    To bound $P_3$, we note that
    \begin{align*}
    \begin{split}
        P_3 &= \sum_{\bm{q}} P(\bm{q}) \cdot P(\bm{q})^2 \\
            &\leq \left(\max_{\bm{q}} P(\bm{q})\right) P_2 \\
            &\leq \left(\max_{\bm{q}}\tr{\rho \left(\frac{1}{2^{s}}\prod_{i \in \bm{q}}\ketbra{\psi_i}{\psi_i}\right)}\right) \left(\frac{1}{3^{s}}\right) \\
            &\leq \frac{1}{6^{s}}.
    \end{split}
    \end{align*}
    Plugging this back into the bound on $\var[\hat{S}^{(K)}]$, we obtain
    \begin{align}
        \var[\hat{S}^{(K)}]\leq \frac{2}{K(K-1)}\left(\frac{1}{3^{s}}+\frac{2(K-2)}{6^{s}}\right).
    \end{align}
    Since $\hat{\mathcal{C}}^{(K)} = 1 - 3^{s} \hat{S}^{(K)}$ gives an unbiased estimate of $\mathcal{C}$, we obtain
    \begin{align}
        \var[\hat{\mathcal{C}}^{(K)}] &= 9^{s} \var[\hat{S}^{(K)}] \nonumber \\
                                      &\leq \frac{2}{K(K-1)}\left(3^{s}+2(K-2)\left(\frac{3}{2}\right)^{s}\right). \label{eqn:CE_SIC_K_var_bound}
    \end{align}
    With the future use of median-of-means estimation in mind, we impose the requirement that $\var[\hat{\mathcal{C}}^{(K)}] \leq \epsilon^2 / 4$, given the precision $\epsilon > 0$ for estimating CE.
    Using this requirement with the bound on the variance of $\hat{\mathcal{C}}^{(K)}$ determined in Eq.~\eqref{eqn:CE_SIC_K_var_bound}, we obtain the inequality
    \begin{equation*}
        K^2 - \left(\frac{16}{\epsilon^2} \left(\frac{3}{2}\right)^s + 1\right) K - \frac{8 \cdot 3^s}{\epsilon^2} \left(1 - \frac{1}{2^{|s| - 1}}\right) \geq 0.
    \end{equation*}
    Noting that $1 - \frac{1}{2^{s - 1}} \geq 0$ for $s \geq 1$, we obtain the solution $K = K_{\mathrm{opt}}$ given in Eq.~\eqref{eqn:SIC_K_opt}, which is the smallest integer value of $K$ satisfying the above inequality.
    Since for any integer $n$, we have $n \leq \lceil n \rceil \leq n + 1$, one can verify using Eq.~\eqref{eqn:CE_SIC_K_var_bound} that the above choice of $K$ gives $\var[\hat{\mathcal{C}}^{(K)}] \leq \epsilon^2/4$.

    Thus, $\hat{\mathcal{C}}_{\ket{\psi}}^{(b)}(S)$ defined in Eq.~\eqref{eqn:SIC_K_opt_CE_estimator} is an unbiased estimator of $\mathcal{C}_{\ket{\psi}}(S)$,
    with $\var[\hat{\mathcal{C}}_{\ket{\psi}}^{(b)}] \leq \epsilon^2/4$ for each $b$.
    Consequently, using Proposition~\ref{prop:med-of-means}, we obtain the desired result.
\end{proof}

The only remaining step is to compute the variance of the estimator for the sum of squared probabilities for SIC measurements.
We summarize this calculation in the lemma below.
\begin{lemma}
    \label{lem:SIC_sumsqprobest_K_var}
    Given iid outcomes $\bm{Q}_1, \dotsc, \bm{Q}_K$ from SIC measurements with $K \geq 2$, let
    \begin{equation}
        \hat{S}^{(K)} = \frac{1}{K (K - 1)} \sum_{\substack{k, k' = 1\\ k \neq k'}}^K \ind{\bm{Q}_{k} = \bm{Q}_{k'}}
    \end{equation}
    be an unbiased estimator of $P_2 = \sum_{\bm{q}} P(\bm{q})^2$, where $P(\bm{q})$ denotes the probability of observing the string $\bm{q}$ after a local SIC measurement.
    Then, the variance of this estimation is given by
    \begin{equation}
        \var[\hat{S}^{(K)}] = \frac{2 P_2 (1 - P_2) + 4 (K - 2) (P_3 - P_2^2)}{K (K - 1)}, \label{eqn:SIC_sumsqprobest_K_var}
    \end{equation}
    where $P_3 = \sum_{\bm{q}} P(\bm{q})^3$.
\end{lemma}
\begin{proof}
    That $\hat{S}^{(K)}$ is an unbiased estimator of $P_2$ can be verified using similar arguments given for LRMs.
    To compute the variance, we again break the sum defining $\hat{S}^{(K)}$ into three cases, following the proof of Lemma~\ref{lem:sumsqprobest_var} for LRMs.
    The first two cases are identical to the variance calculation for LRMs, except that $P_2 = \sum_{\bm{q}} P(\bm{q})^2$ and $P_3 = \sum_{\bm{q}} P(\bm{q})^3$.
    The third case is also the same, except that $\cov(X_{k, k'}, X_{j, j'}) = \mathbb{E}[X_{k, k'} X_{j, j'}] - \mathbb{E}[X_{k, k'}] \mathbb{E}[X_{j, j'}] = 0$,
    since $X_{k, k'}$ and $X_{j, j'}$ are independent with respect to the outcome probabilities and we do not need to take an expectation over unitaries for SIC measurements.
    Putting these observations together, we find that the variance of $\hat{S}^{(K)}$ is given by Eq.~\eqref{eqn:SIC_sumsqprobest_K_var}.
\end{proof}

\section{Simulation Details}
\subsection{LRMs Simulations}\label{app:LRM-Simulations}
In the main text, our simulations evaluate the CEs of two canonical examples of entangled states. First, the $n$-qubit Greenberger-Horne-Zeilinger (GHZ) state, which is defined as 
\begin{align}
    \ket{\mathrm{GHZ}}_n = \frac{1}{\sqrt{2}}(\ket{0}^{\otimes n} + \ket{1}^{\otimes n}  ).
\end{align}
Because the purity of any subsystem of this state is $1/2$, the CE of the full state can easily be computed as \cite{beckey2023Multipartite}
\begin{align}
    \CC_{\ket{\mathrm{GHZ}}}([n]) = \frac{1}{2} - \frac{1}{2^n}.
\end{align}
The W state is defined as a uniform superposition of all states with Hamming weight one
\begin{align}
    \ket{W}_n = \frac{1}{\sqrt{n}}(\ket{10\dotsm 0} + \ket{010\dotsm 0} + \dotsm + \ket{0\dotsm 01}).
\end{align}
The general formula is somewhat harder to prove, but the result is similarly simple \cite{beckey2023Multipartite}
\begin{align}
    \CC_{\ket{W}}([n]) = \frac{1}{2} - \frac{1}{2n}.
\end{align}
These are the exact formulas used to benchmark the performance of our estimators as shown in Fig.~\ref{fig:CE-LRM-GHZ-W}.

Recall from Theorem~\ref{thm:Ohnemus-LRM-Hoeffding} the estimator for the CEs based on $L$ rounds of LRMs, $\bm{Z}_{1,1},\ldots,\bm{Z}_{1,K},\ldots,\bm{Z}_{L,1},\ldots,\bm{Z}_{L,K}$, is,
\begin{align}
    \hat{\CC}_{\ket{\psi}}(S)=1-\left(\frac{3}{2}\right)^{s}\frac{1}{L}\sum_{l=1}^L \hat{S}^{(K)}_l.
\end{align}
From this, we employ the following algorithm to simulate this estimator:
\begin{algorithm}[H]
    \caption{LRM-Mean Simulation}\label{algo:MeanLRM}
        \begin{algorithmic}[1]
        \State \textbf{Input:} $n$-qubit pure state $\rho = \ketbra{\psi}{\psi}$, $S \subseteq [n]$, number of unitaries $L$ of the form $\prod_{i \in S}U_i$, and $K\geq 2$ shots per unitary.
        \State Let $\DC_l$ denote a discrete probability distribution over bitstrings and $C(2)$ denote the set of single qubit Clifford gates.
        \Function{MeanCEviaLRMs}{$\rho,S,L,K$}
        \State $l=1$
        \State $\hat{S}^{(K)}=0$
        \While {$l \leq L$}
            \State $U_l = \prod_{i \in S} U_i ,\quad U_i \in C(2)$
            \State $\DC_l=\{\bra{\bm{z}}U_l^\dagger \rho U_l \ket{\bm{z}}\rightarrow \bm{z}\}_{\bm{z}=0}^{2^s-1}$
            \State $M_l=\mathrm{Sample}(\DC_l,K)$
            \For{$i \neq j \in M_l$}
                \If{$\bm{Z}_{l,i}=\bm{Z}_{l,j}$}
                    $\hat{S}^{(K)} = \hat{S}^{(K)} +1$
                \EndIf
            \EndFor
            \State $l=l+1$
        \EndWhile
        \State \textbf{return} $\hat{\CC}_{\ket{\psi}}(S)=1-\left(\frac{3}{2}\right)^{s} \frac{1}{L}\frac{1}{K(K-1)} \hat{S}^{(K)}$
        \EndFunction
        \end{algorithmic}
\end{algorithm}

It is worth noting that for Fig.~\ref{fig:CE-LRM-GHZ-W} unitaries were sampled from the Haar distribution rather than from the single qubit Cliffords. In practice, it is likely single qubit Cliffords would be preferable. 

For median-of-means estimation, we directly call this function with a specific value for $L$:

\begin{algorithm}[H]
    \caption{LRM-MoM Simulation}\label{algo:MoMLRM}
        \begin{algorithmic}[1]
        \State \textbf{Input:} $n$-qubit pure state $\rho =\ketbra{\psi}{\psi}$, $S\subseteq [n]$, number of unitaries $L$, shots per unitary $K\geq 2$, and confidence level $1-\delta\in(0,1)$.
        \Function{MoMCEviaLRMs}{$\rho,S,L,K,\delta$}
        \State $N_B = \lceil 8 \log(\frac{1}{\delta}) \rceil$
        \State $B = \lceil \frac{L}{N_B}\rceil$
        \State \textbf{return} $\mathrm{Median}[\{\mathrm{MeanCEviaLRMs}[\rho,S,B,K]\}_{i=1}^{N_B}]$
        \EndFunction
        \end{algorithmic}
\end{algorithm}

\subsection{SICs Simulations}
Following the same modality as in the LRM section, we have the following estimator from Theorem~\ref{thm:CEs-via-SICs-MoM-K=2},
\begin{align}
        \hat{\CC}_{\ket{\psi}}(S)=1-3^s\hat{S}^{(K)},
\end{align}
where $\bm{Q}_1,\ldots,\bm{Q}_M$ are the outcomes from $M$ rounds of SIC measurements. For ease of notation, let $\rho \otimes \ketbra{0}{0}^{\otimes n}$ denote the state with qubits ordered such that each system qubit is next to its ancilla qubit. Then the following algorithm simulates this estimator:\\

\begin{algorithm}[H]
    \caption{SIC-MoM Simulation}\label{algo:MoMSIC}
        \begin{algorithmic}[1]
        \State \textbf{Input:} $n$-qubit pure state $\rho = \ketbra{\psi}{\psi}$, $S\subseteq [n]$, $K$, and $\delta>0$ confidence level.
        \State Let $\DC$ denote a discrete probability distribution.
        \Function{MoMCEviaSICs}{$\rho,S,K,\delta$}
        \State $N_B$=$\lceil 8 \log (\frac{1}{\delta})\rceil$
        \State $U = \prod_{i \in S}U_\mathrm{SIC}$
        \State $\DC=\{\bra{\bm{q}}U^\dagger(\rho \otimes \ketbra{0}{0}^{\otimes n}) U\ket{\bm{q}}\rightarrow \bm{q}\}_{\bm{q}=0}^{4^{s}-1}$
        \State $b=1$
        \While{$b\leq N_B$}
        \State $\hat{S}^{(K)}_b=0$
        \State $M_b=\mathrm{Sample}(\DC,K)$
            \For{$i\neq j \in M_b$}
                \If{$\bm{Q}_i=\bm{Q}_j$}
                $\hat{S}^{(K)}_b=\hat{S}^{(K)}_b+1$
                \EndIf
            \EndFor
        \State $\overline{\CC}^{(b)}_{\ket{\psi}}(S)=1-3^s \frac{1}{K(K-1)}\hat{S}^{(K)}_b$
        \EndWhile
        \State \textbf{return }$\hat{\CC}_{\ket{\psi}}(S)=\mathrm{median}[\{\overline{\CC}_{\ket{\psi}}^{(b)}(S)\}_{b=1}^{N_B}]$
        \EndFunction
        \end{algorithmic}
\end{algorithm}

To reproduce Fig.~\ref{fig:4histComp-GHZ5}, to each of the simulations outlined in Table~\ref{table:complexities} we input a $5$-qubit GHZ state that is $\rho=\ketbra{\mathrm{GHZ}_5}{\mathrm{GHZ}_5}$ and set $\delta=0.05$ for methods based on median-of-means. The total measurement budget is fixed based on the measurements required for SIC-MoM, $K=K_{\mathrm{opt}}$ to have $\epsilon=0.05$ precision at $1-\delta=0.95$ confidence. That is a total measurement budget of $M=K_{\mathrm{opt}}\lceil 8 \log(\frac{1}{\delta}) \rceil$ where,
\begin{align}
\begin{split}
    K_{\mathrm{opt}} &= \Biggl \lceil \frac{1}{2} \left(\frac{16}{\epsilon^2} \left(\frac{3}{2}\right)^{s} + 1\right) \\
    &+\frac{1}{2} \sqrt{\left(\frac{16}{\epsilon^2} \left(\frac{3}{2}\right)^{s} - 1\right)^2 + \frac{32}{\epsilon^2} 3^{s}} \Biggr \rceil.
\end{split}
\end{align}
The corresponding measurement splits were then computed from this total number. Then $1,000$ trials of $M$ total measurements of LRM data and SIC data were generated. Since both LRM-Mean and LRM-MoM have $K=2$, i.e. two shots per unitary, the same data was post-processed according to each estimator. This is possible as median-of-means only changes how things are batched since the total measurement budget was fixed. The same was done for the SIC estimators. As a final comment, to begin to see the asymptotic behavior within this plot, one would have to go up to roughly 12-13 qubits. This is because of the various pre-factors hidden by big-$O$ notation.

\end{document}